\documentclass[10pt]{article}
\usepackage[T1]{fontenc}
\usepackage[latin9]{inputenc}
\usepackage{geometry}
\usepackage{color}
\usepackage{mathrsfs}
\usepackage{amsmath}
\usepackage{amsthm}
\usepackage{amssymb}
\usepackage{graphicx}
\usepackage{setspace}
\usepackage[authoryear]{natbib}
\usepackage[unicode=true,
 bookmarks=false,
 breaklinks=false,pdfborder={0 0 1},backref=section,colorlinks=true]
 {hyperref}
\hypersetup{
 linkcolor=red,citecolor=blue,filecolor=magneta,urlcolor=cyan}

\makeatletter

\providecommand{\tabularnewline}{\\}
\newcommand{\lyxdot}{.}

\theoremstyle{plain}
\newtheorem{thm}{\protect\theoremname}

\date{}
\usepackage{array}
\usepackage{amsfonts}
\usepackage{rotate}
\usepackage{epsfig}\usepackage{amsfonts}
\usepackage{lscape}
\usepackage{color}
\usepackage{longtable}
\graphicspath{ {} }

\usepackage[all]{xy}
\usepackage{amsthm}\usepackage{tikz}\usepackage{fancyhdr}\usepackage{fancyhdr}\usepackage{mathrsfs}\usepackage{color}\usepackage{xspace}\usepackage{pdflscape}
\DeclareMathSizes{1}{5}{5}{5}
\makeatother

\providecommand{\theoremname}{Theorem}

\begin{document}
\title{Statistical Inference from Partially Nominated Sets: An Application
to Estimating the Prevalence of Osteoporosis Among Adult Women}
\author{Zeinab Akbari Ghamsari, Ehsan Zamanzade\thanks{Corresponding author. e.zamanzade@sci.ui.ac.ir; ehsanzamanzadeh@yahoo.com}
 \space and Majid Asadi}

\maketitle
Department of Statistics, Faculty of Mathematics and Statistics, University
of Isfahan,

Isfahan 81746-73441, Iran
\begin{abstract}
This paper focuses on drawing statistical inference based on a novel
variant of maxima or minima nomination sampling (NS) designs. These
sampling designs are useful for obtaining more representative sample
units from the tails of the population distribution using the available
auxiliary ranking information. However, one common difficulty in performing
NS in practice is that the researcher cannot obtain a nominated sample
unless he/she uniquely determines the sample unit with the highest
or the lowest rank in each set. To overcome this problem, a variant
of NS which is called partial nomination sampling is proposed in which
the researcher is allowed to declare that two or more units are tied
in the ranks whenever he/she cannot find with high confidence the
sample unit with the highest or the lowest rank with high confidence.
Based on this sampling design, two asymptotically unbiased estimators
are developed for the cumulative distribution function, which are
obtained using maximum likelihood and moment-based approaches, and
their asymptotic normality is proved. Several numerical studies have
shown that the developed estimators have higher relative efficiencies
than their counterpart in simple random sampling in analyzing either
the upper or the lower tail of the parent distribution. The procedures
that we are developed are then implemented on a real dataset from
the Third National Health and Nutrition Examination Survey (NHANES
III) to estimate the prevalence of osteoporosis among adult women
aged 50 and over. It is shown that in some certain circumstances,
the techniques that we have developed require only one-third of the
sample size needed in SRS to achieve the desired precision. This results
in a considerable reduction of time and cost compared to the standard
SRS method.


\begin{singlespace}
\end{singlespace}
\end{abstract}
\thispagestyle{empty}

\noindent \newtheorem{theorem}{Theorem}
\newtheorem{remark}{Remark}

\noindent \setcounter{page}{2}

\newpage{}

\section{Introduction\label{Intro}}

In some medical applications such as survival analysis, cancer studies,
risk analysis, and prevalence studies, the researcher is interested
in drawing a statistical inference about either the lower or the upper
tail of the population distribution. For instance, in large cohort
studies on measuring bone mineral density (BMD) among elderly women
in a given population, the researcher is often interested in the lower
quantile of the population since the lower values of BMD are related
to a higher possibility of bone fracture due to some minor incidents
such as falls or sport injuries. This imposes a substantial financial
burden on the society due to the high cost of medical treatment and
the possibility of the patients\textquoteright{} lifetime disability.

In situations where the researcher is only interested in either the
upper or the lower tail of the population distribution, the usual
simple random sampling (SRS) design might not lead to a good representative
sample since enough units from the tails of the distribution may not
be included in the sample. On the other hand, when measuring the actual
value of the variable of interest is time-consuming, expensive, or
hazardous, a large enough SRS sample may not be available to guarantee
that the resulting sample captures all aspects of the given population
related to at least one of the tails of the distribution. There is
sometimes a wealth of auxiliary information about the population of
interest which cannot be translated into covariates. Nevertheless,
they can be used by some more structured sampling designs such as
nomination sampling (NS) to select a more representative sample from
the tails of population as is possible in SRS. For instance, assume
that a medical researcher is interested in obtaining the prevalence
of women with a BMD lower than a given threshold in a certain population.
Note that measuring a person's actual BMD is an expensive and time-consuming
job because it requires Dual-energy X-ray absorptiometry (DXA) technology
which is expensive and not easily accessible in developing countries.
In addition, it is inconvenient since it needs a medical expert to
segment the pertinent images manually. However, a medic can first
use his/her past experience to nominate the patients who are more
likely to have the lowest BMD values among the others for further
consideration and then the actual BMD values of the nominated sample
units can be obtained via DXA technology. As has been shown in the
literature (for example, see \citealp{Zamanzade=000026Mahi2020}),
the use of NS rather than SRS can lead to a considerable efficiency
gain at either the lower or the upper tail of the distribution. As
a result, a substantial reduction in time and cost will be achieved
for a predefined precision.

The NS technique was first suggested by \citet{Willemain-a} in his
effort to develop a more acceptable way to pay for nursing home services.
He argued that nursing home service operators would be reluctant to
accept sample-based reimbursement rate by the government unless they
were allowed to play a role in the sampling process. In particular,
the operators might be allowed to include those residents with the
highest caring costs in the sample. Hence, the cited author introduced
NS to draw a statistical inference based on such samples. This sampling
strategy is also the only practical option when historical data include
only extreme values \citep{Willemain-b}. NS has been used for estimating
the population median \citep{Willemain-b}, the cumulative distribution
function (CDF) \citep{Boyles}, constructing acceptance sampling plans
\citep{Jozani-Mirkamali1}, and the quality control charts \citep{Jozani-Mirkamali2},
quantile estimation \citep{Nourmohammadi2014} and quantile regression
\citep{Jozani-NSreg}.

NS can be regarded as a special case of ranked set sampling (RSS)
proposed by \citet{McIntyre} as an efficient method for estimating
pasture yields in Australia. RSS is applicable in situations where
the sample units in a small-sized set can be easily ranked without
referring to their actual values. These situations may occur in environmental,
ecological, biological, medical, and fisheries studies. In fact, the
sample units in RSS are obtained based on additional information from
some unquantified units in terms of judgment ranks in order to provide
a more structured sample from the population of interest. Since its
introduction, RSS has been the subject of many studies and virtually
all standard statistical problems using RSS have been addressed in
the literature. These include, among others, the estimation of the
population mean \citep{Takahasi,Ozturk2011,Frey2012,FreyMean}, the
CDF \citep{Stokes,Samawi-CDF-2001,Lutz}, the population proportion
\citep{Chen2005,Zamanzade2017,Omidvar,Frey2019,Frey2021}, perfect
ranking tests \citep{Frey2007,Frey.perfect1,Frey=000026Wang,Frey-EDF},
odds ratio \citep{Samawi-odds-2013}, logistic regression \citep{Samawi-odds-2013,Samawi-logi-2017,Samawi-sample-2018,Samawi-sample-2020},
the Youden index \citep{Samawi-StatMed-2017}, parameter estimation
\citep{Chen2018,Chen2019,He2020,He2021,Qian2021}, randomized cluster
design \citep{Ahn,Wang2016,Wang2017,Wang2020} and statistical control
quality charts \citep{Al-omari,Haq,Haq-almori1,Haqetal2}.

In the NS design, a sample unit with the highest/lowest judgment rank
in a small-sized set has to be determined uniquely by the researcher.
In this paper, a new variant of NS called \textquoteleft partial nomination
sampling (PNS) design\textquoteright{} is developed in which the researcher
is allowed to declare ties whenever he/she cannot find with high confidence
the sample unit with the highest/lowest rank. In Section \ref{Sampling design},
PNS is described in detail and its usage is shown by a hypothetical
example. In Section \ref{sec:CDF-estimation-using}, two estimators
of the CDF are developed and their asymptotic normality is proved.
In Section \ref{sec:Comparisons}, the estimators that we developed
are compared with their SRS counterpart. A real dataset is analyzed
in Section \ref{sec:Data-Analysis-using} to show the applicability
of our proposed procedure in the current paper. Some concluding remarks
and directions for future research are provided in Section \ref{sec:Concluding-Remarks}.

\section{Sampling design set up\label{Sampling design}}

Let $m$ be the set size. The NS proposed by \citet{Willemain-a}
is a two-step sampling method. In the first step, the researcher is
required to draw $n$ simple random samples (each with the size of
$m$) from an infinite population and then rank each sample with the
size of $m$ from the smallest to the largest with respect to the
variable of interest using any relatively cheap and convenient method
which does not entail the actual quantification of the sample units.
In the second step, the researcher draws one unit with the smallest/largest
rank from each set for exact measurement. Therefore, this process
can be outlined as follows:
\begin{enumerate}
\item Determine an integer value for the size of the set $\left(m\right)$.
In order to facilitate the ranking process, the value of $m$ should
be kept small (e.g. from 2 to 10).
\item Identify $n\times m$ units from the target population and partition
them randomly into $n$ sets each with the size of $m$. Then, rank
each set with the size of $m$ from the smallest to the largest. The
ranking process in this step is done without any reference to the
actual values of the set units. Therefore, it is prone to error (imperfect).
\item Select the sample unit with the smallest rank in each set with the
size of $m$ for actual measurement.
\end{enumerate}
The above sampling design is called \textquoteleft MinNS\textquoteright{}
and the sample units are denoted by $\left\{ Y_{\left[1\right]i},i=1,\ldots n\right\} ,$
where $Y_{\left[1\right]i}$ denotes the sample unit with the smallest
judgment rank in the $i$-th set. In the third step mentioned above,
the sample unit with the highest judgment rank in each set with the
size of $m$ is selected for actual measurement. In this case, the
sampling design is called \textquoteleft MaxNS\textquoteright{} and
the resulting sample units are denoted by $\left\{ Y_{\left[m\right]i},i=1,\ldots n\right\} ,$
where $Y_{\left[m\right]i}$ denotes the sample unit with the highest
judgment rank in the $i$-th set. There are two main difficulties
in performing NS in practice which are both rooted in the nature of
the ranking process since it is done without any reference to the
actual values of the sample units. The first one is called \textquoteleft imperfect
ranking\textquoteright{} which refers to situations in which the sample
unit with the lowest (highest) judgment rank in the set with the size
of $m$ may not be the same as the sample unit with the true lowest
(highest) rank in the set. Therefore, it is essential to show that
the developed procedures based on the NS design are relatively robust
for imperfect ranking as long as the quality of ranking is fairly
good. The second hardship in applying NS in real situations is about
the ties in the ranking process which pertains to situations in which
the researcher cannot determine with high confidence the sample unit
with the lowest (highest) judgment rank in the set with the size of
$m$ and therefore two or more units in the set are tied as the lowest
(highest) rank. The current solution for this problem is to ignore
the information of the ties and select one of the tied units at random
for actual quantification. In this paper, not only breaking the ties
at random in NS but also recording their information for use in the
estimation process are proposed. The outline of the proposed sampling
method is given below:
\begin{enumerate}
\item Denote by $\mathbf{I}$, the matrix of the ties information with $n$
rows and $m$ columns, where $m$ is the size of the set and $n$
is the sample size.
\item Similar to NS, identify $n\times m$ units from the target population
and randomly partition them into $n$ sets each with the size of $m$
and then rank each set from the smallest to the largest.
\item From the $i$-th set, select the sample unit with the lowest judgment
rank (for $i=1,\ldots,n$). Define $I_{i,j}=1$, if the researcher
cannot distinguish between the sample unit with the lowest judgment
rank in the $i$-th set and the sample unit with the judgment rank
$j$ in the same set. Otherwise, $I_{i,j}=0$, where $I_{i,j}$ is
the element in the $i$-th row and $j$-th column of the matrix of
the ties information $\mathbf{I}$.
\end{enumerate}
The above sampling plan is called \textquoteleft MinPNS\textquoteright{}
and the sample units include not only the measured values of $\left\{ Y_{\left[1\right]i},i=1,\ldots n\right\} ,$
but also the matrix of the ties information

\noindent$\left(\mathbf{I}=\left\{ I_{i,j},i=1,\ldots,n;j=1,\ldots,m\right\} \right).$
Since in MinPNS, the sample unit with the minimum judgment rank is
always tied to itself, we have $I_{i,1}=1$ for $i=1,\ldots,n$, and
therefore $\sum_{j=1}^{m}I_{i,j}\geq1$. Clearly, MinPNS is reduced
to MinNS if and only if $\sum_{j=1}^{m}I_{i,j}=1$. If in the third
step of the above sampling plan, the researcher is concerned with
the sample unit with the highest rank in each set, then the sampling
design is called \textquoteleft MaxPNS\textquoteright . Note that
in the MaxPNS design, we also have $\sum_{j=1}^{m}I_{i,j}\geq1$ because
the sample unit with the maximum rank is always tied to itself and
also MaxPNS is equivalent to MaxNS if and only if $\sum_{j=1}^{m}I_{i,j}=1$.
It is important to note that the sample units in the PNS plan are
independent but not identically distributed. Let $m_{i}=\sum_{j=1}^{m}I_{i,j},$
for $i=1,\ldots,n$ and $n_{r}=\sum_{i=1}^{n}\mathbb{I}\left(m_{i}=r\right),$
for $r=1,\ldots,m$, where $\mathbb{I}\left(.\right)$ is the indicator
function. Then, in the MinPNS (MaxPNS) design, $n_{r}$ is the number
of the measured units in PNS which are tied to $r$ units with the
smallest (largest) ranks in the set. It is clear that $n=\sum_{r=1}^{m}n_{r}$.
With this definition, the measured sample units which fall in the
$r$-th stratum are independent and identically distributed.

The notions of perfect and imperfect rankings in PNS need to be slightly
adjusted. In the MinPNS (MaxPNS) design, the ranking is called perfect
if all the tied units in each set with the size of $m$ are smaller
(larger) than the other units in the set. Note that in the MinPNS
(MaxPNS) design, in the case of perfect ranking, the smallest (largest)
tied units are not ranked and one of them is selected at random for
actual measurement. \textquoteleft Imperfect ranking\textquoteright{}
in MinPNS (MaxPNS) refers to a situation in which at least one of
the tied units is larger (smaller) than the other untied units in
the set with the size of $m$. Clearly, if there are no ties in the
ranking process, these definitions of perfect and imperfect rankings
in PNS will coincide with their counterparts in NS.

In what follows, the MinPNS plan is illustrated by providing a hypothetical
example. Suppose that a medical researcher is interested in drawing
a statistical inference about the lower tail of the distribution of
the patients\textquoteright{} BMD values in a given population. Since
the exact measurement of BMD values is much harder than ranking them
in a small-sized set, MinPNS seems to be an appealing alternative
to the usual SRS. To demonstrate a MinPNS design with the size of
$n=10$ and the set size of $m=5,$ $m\times n=50$ patients are randomly
selected from the population and divided into $n=10$ sets each with
the size of $m=5.$ In each set, the patients are judgmentally ranked
according to their BMD values using any inexpensive method such as
the patient's check-up documents, determining their body mass index
(BMI), or the researcher's personal experience. Then, the patient
whose BMD is most likely the lowest is selected for actual quantification
using DXA technology. Whenever the researcher cannot determine with
high confidence the patient with the lowest BMD in a set, he/she is
allowed to declare as many ties as needed and to select one of the
tied patients at random for actual quantification. An example of MinPNS
is shown in detail in Table \ref{table1} where each row corresponds
to one set and the sample units in each set are listed according to
the judgment rank of their BMD values after the possible ties are
broken at random. Let $u_{i,j}$ be the patient with the $j$-th judgment
rank in the $i$-th set with the size of $m=5$ ($i\in\left\lbrace 1,\ldots,10\right\rbrace $
and $j\in\left\lbrace 1,\ldots,5\right\rbrace $). In Table \ref{table1},
the tied units are denoted by ampersands and the selected unit for
actual quantification is shown in bold face type.

\begin{table}
\centering %
\begin{tabular}{cccc}
\hline 
Set & Ranked units & Selected Unit for measurement & Actual BMD value ($Y$)\tabularnewline
\hline 
1 & $\boldsymbol{u_{1,1}},u_{1,2},u_{1,3},u_{1,4},u_{1,5}$ & $u_{1,1}$ & $0.884$\tabularnewline
2 & $u_{2,1}\&\boldsymbol{u_{2,2}},u_{2,3},u_{2,4},u_{2,5}$ & $u_{2,2}$ & $0.610$\tabularnewline
3 & $\boldsymbol{u_{3,1}},u_{3,2},u_{3,3}\&u_{3,4},u_{3,5}$ & $u_{3,1}$ & $0.753$\tabularnewline
4 & $u_{4,1}\&\boldsymbol{u_{4,2}}\&u_{4,3}\&u_{4,4}\&u_{4,5}$ & $u_{4,2}$ & $0.616$\tabularnewline
5 & $u_{5,1}\&u_{5,2}\&\boldsymbol{u_{5,3}}\&u_{5,4},u_{5,5}$ & $u_{5,3}$ & $0.690$\tabularnewline
6 & $\boldsymbol{u_{6,1}},u_{6,2},u_{6,3},u_{6,4},u_{6,5}$ & $u_{6,1}$ & $0.542$\tabularnewline
7 & $u_{7,1}\&u_{7,2}\&\boldsymbol{u_{7,3}},u_{7,4},u_{7,5}$ & $u_{7,3}$ & $0.576$\tabularnewline
8 & $u_{8,1}\&u_{8,2}\&\boldsymbol{u_{8,3}},u_{8,4},u_{8,5}$ & $u_{8,3}$ & $0.698$\tabularnewline
9 & $\boldsymbol{u_{9,1}},u_{9,2},u_{9,3}\&u_{9,4}\&u_{9,5}$ & $u_{9,1}$ & $0.769$\tabularnewline
10 & $u_{10,1}\&u_{10,2}\&\boldsymbol{u_{10,3}},u_{10,4},u_{10,5}$ & $u_{10,3}$ & $0.670$\tabularnewline
\hline 
\end{tabular}\caption{\label{table1}{\small{}An example of MinPNS with $n=10$ and $m=5$.
In each set, the sample units are listed according to the judgment
rank of their actual BMD values}}
\end{table}

One can observe in Table \ref{table1} that in the first set, the
researcher is able to determine with high confidence the sample unit
with the lowest rank and therefore no tie is declared. In the second
set, the researcher cannot determine which of the two sample units
has the lowest rank. Therefore, he/she selects one of them at random
and records the ties information. Although there is a tie in the third
set, it is not recorded because it is not related to the sample unit
with the lowest rank and is therefore irrelevant. In the fourth set,
the researcher is not able to have any judgment ranking. Therefore,
one patient is randomly selected from the set. Other sample units
are obtained in a similar fashion. Thus, the information matrix of
the ties is given by:
\[
\mathbf{I}=\begin{bmatrix}1 & 0 & 0 & 0 & 0\\
1 & 1 & 0 & 0 & 0\\
1 & 0 & 0 & 0 & 0\\
1 & 1 & 1 & 1 & 1\\
1 & 1 & 1 & 1 & 0\\
1 & 0 & 0 & 0 & 0\\
1 & 1 & 1 & 0 & 0\\
1 & 1 & 1 & 0 & 0\\
1 & 0 & 0 & 0 & 0\\
1 & 1 & 1 & 0 & 0
\end{bmatrix}
\]

In Table \ref{table1} and the information matrix of the ties $\left(\boldsymbol{\mathbf{I}}\right)$,
we see that, $u_{1,1},u_{3,1},u_{6,1}$ and $u_{9,1}$ fall in the
first stratum, $u_{2,2}$ falls in the second stratum, $u_{7,3},u_{8,3}$
and $u_{10,3}$ fall in the third stratum, and $u_{5,3}$ and $u_{4,2}$,
respectively fall in the fourth and fifth strata. Therefore, the vector
of $\mathbf{n}=\left(n_{1},\ldots,n_{5}\right)$ is given by $\mathbf{n}=\left(4,1,3,1,1\right)$.

\section{CDF estimation using PNS\label{sec:CDF-estimation-using}}

In this section, two CDF estimators are developed using the PNS design
and their asymptotic performance is studied. For brevity, only the
results for MinPNS are presented. The results for MaxPNS can be obtained
in a similar fashion. Let $\left\{ Y_{\left[1\right]i},i=1,\ldots n\right\} ,$
and $\mathbf{I}=\left\{ I_{i,j},i=1,\ldots,n;j=1,\ldots,m\right\} $
be the sample units and the information matrix of the ties, respectively,
obtained using the MinPNS design from a population with the CDF of
$F$. Let $Y_{\left[r\right]}^{j}$ be the $j$-th measured unit in
the $r$-th stratum (for $j=1,\ldots,n_{r};r=1,\ldots m$). Note that,
alternatively, the sample units from the MinPNS design can be represented
as $\left\{ Y_{\left[r\right]}^{j},j=1,\ldots,n_{r};r=1,\ldots m\right\} $.
Here, the square brackets $\left[.\right]$ are used to show that
the ranking might be imperfect and thus prone to error. If there is
no error in ranking (perfect ranking), then square brackets $\left[.\right]$
are replaced by round ones $\left(.\right)$ and the MinPNS sample
units are denoted by $\left\{ Y_{\left(r\right)}^{j},j=1,\ldots,n_{r};r=1,\ldots m\right\} $.
In what follows, two different approaches for estimating the CDF of
$F$ based on moment and maximum likelihood (ML) techniques are described
one by one.

\subsection{Moment-based (MB) approach}

The well-known naive (empirical) estimator of the CDF is given by:

\begin{equation}
F_{n}\left(t\right)=\frac{1}{n}\sum_{i=1}^{n}\mathbb{I}\left(Y_{\left[1\right]i}\leq t\right),
\end{equation}

where $\mathbb{I}\left(.\right)$ is the indicator function.

In order to obtain the asymptotic distribution of $F_{n}\left(t\right)$,
note that

\begin{align*}
F_{n}\left(t\right) & =\frac{1}{n}\sum_{i=1}^{n}\mathbb{I}\left(Y_{\left[1\right]i}\leq t\right)=\frac{1}{n}\sum_{r=1}^{m}\sum_{j=1}^{n_{r}}\mathbb{I}\left(Y_{\left[r\right]}^{j}\leq t\right)\\
 & =\sum_{r=1}^{m}\frac{n_{r}}{n}\times\frac{1}{n_{r}}\sum_{j=1}^{n_{r}}\mathbb{I}\left(Y_{\left[r\right]}^{j}\leq t\right)=\sum_{r=1}^{m}q_{r}F_{n,\left[r\right]}\left(t\right),
\end{align*}

where $q_{r}=\frac{n_{r}}{n}$, and $F_{n,\left[r\right]}\left(t\right)=\frac{1}{n_{r}}\sum_{j=1}^{n_{r}}\mathbb{I}\left(Y_{\left[r\right]}^{j}\leq t\right)$
is the empirical distribution function (EDF) based on the sample units
in the $r$-th stratum (for $r=1,\ldots,m$). Now, if we let $n_{r}\rightarrow+\infty,$
such that $n=\sum_{r=1}^{m}n_{r}\rightarrow+\infty$, $\frac{n_{r}}{n}\rightarrow\lambda_{r}\in\left(0,1\right)$,
and $\sum_{r=1}^{m}\lambda_{r}=1$, then, it follows from the central
limit theorem (CLT) that $\sqrt{n_{r}}\left(F_{n,\left[r\right]}\left(t\right)-F_{\left[r\right]}\left(t\right)\right)\xrightarrow{d}N\left(0,F_{\left[r\right]}\left(t\right)\left(1-F_{\left[r\right]}\left(t\right)\right)\right)$,
where $\xrightarrow{d}$ denotes convergence in distribution and $F_{\left[r\right]}\left(t\right)$
is the CDF of the sample units in the $r$-th stratum. Therefore,
one can write:

\begin{align*}
\sqrt{n}\left(F_{n}\left(t\right)-F_{\lambda}\left(t\right)\right) & =\sqrt{n}\left(\sum_{r=1}^{m}\left(q_{r}F_{n,\left[r\right]}\left(t\right)-\lambda_{r}F_{\left[r\right]}\left(t\right)\right)\right)\\
 & =\left(\sqrt{n}\sum_{r=1}^{m}q_{r}\left(F_{n,\left[r\right]}\left(t\right)-F_{\left[r\right]}\left(t\right)\right)\right)+\sqrt{n}\left(\sum_{r=1}^{m}F_{\left[r\right]}\left(t\right)\left(q_{r}-\lambda_{r}\right)\right)\\
 & =\sum_{r=1}^{m}\left(\sqrt{q_{r}}\times\sqrt{n_{r}}\left(F_{n,\left[r\right]}\left(t\right)-F_{\left[r\right]}\left(t\right)\right)\right)+\sqrt{n}\left(\sum_{r=1}^{m}F_{\left[r\right]}\left(t\right)\left(q_{r}-\lambda_{r}\right)\right)\\
 & \xrightarrow{d}N\left(0,\sum_{r=1}^{m}\lambda_{r}F_{\left[r\right]}\left(t\right)\left(1-F_{\left[r\right]}\left(t\right)\right)\right),
\end{align*}

where $F_{\lambda}\left(t\right)=\sum_{r=1}^{m}\lambda_{r}F_{\left[r\right]}\left(t\right)$
and the last convergence follows Slutsky's theorem. The derivation
of the moment-based estimator requires a perfect ranking assumption.
Assume that the ranking is perfect, then one can simply show that
the CDF $Y_{\left(r\right)}^{j}$ is given by:

\[
F_{\left(r\right)}\left(t\right)=\frac{1}{r}\sum_{i=1}^{r}\mathbb{B}\left(F\left(t\right),i,m+1-i\right),
\]

where

\[
\mathbb{B}\left(x,i,m+1-i\right)=\int_{0}^{x}i\binom{m}{i}x^{i-1}\left(1-x\right)^{m-i}dx,
\]

is the CDF of the beta distribution with parameters $i$ and $m+1-i$,
evaluated at point $x$. Therefore, the expectation of $F_{n}\left(t\right)$
under the perfect ranking assumption can be obtained as:

\begin{align*}
\mathbb{E}\left(F_{n}\left(t\right)\right) & =\sum_{r=1}^{m}q_{r}\mathbb{E}\left(F_{n,r}\left(t\right)\right)=\sum_{r=1}^{m}q_{r}F_{\left(r\right)}\left(t\right)\\
 & =\sum_{r=1}^{m}\left(\frac{1}{r}\sum_{i=1}^{r}\mathbb{B}\left(F\left(t\right),i,m+1-i\right)\right)=g\left(F\left(t\right)\right),
\end{align*}

where $g\left(x\right)=\sum_{r=1}^{m}\left(\frac{1}{r}\sum_{i=1}^{r}\mathbb{B}\left(x,i,m+1-i\right)\right)$
is a bijective function. Thus, the moment-based estimator of the CDF
of the population $\left(F\left(t\right)\right)$ can be obtained
in the following way:

\[
F_{mb}\left(t\right)=g^{-1}\left(F_{n}\left(t\right)\right),
\]

where $g^{-1}\left(x\right)$ is the inverse of $g\left(x\right)$.

The next theorem establishes the asymptotic normality of $F_{mb}\left(t\right)$.
\begin{thm}
Let $\left\{ Y_{\left(r\right)}^{j},j=1,\ldots,n_{r};r=1,\ldots m\right\} $
be a MinPNS sample from a population with CDF $F\left(t\right)$ which
is obtained under the perfect ranking assumption. Then, $\sqrt{n}\left(F_{mb}\left(t\right)-F(t)\right)$
converges in distribution to a mean zero normal distribution with
variance 
\[
\sigma_{mb}^{2}=\left(\sum_{r=1}^{m}\lambda_{r}F_{\left(r\right)}\left(t\right)\left(1-F_{\left(r\right)}\left(t\right)\right)\right)\left(g^{'}\left(F\left(t\right)\right)\right)^{-2},
\]

when $n^{*}=\min_{r=1}^{m}n_{r}$ goes to infinity.
\end{thm}

\begin{proof}
Note that:

\begin{align*}
\sqrt{n}\left(F_{mb}\left(t\right)-F(t)\right) & =\sqrt{n}\left(g^{-1}\left(F_{n}\left(t\right)\right)-g^{-1}\left(g\left(F(t)\right)\right)\right)\\
 & =\sqrt{n}\left(g^{-1}\right)^{'}\left(g\left(F(t)\right)\right)\left(F_{n}\left(t\right)-g\left(F(t)\right)\right)+O_{p}\left(n^{-\frac{1}{2}}\right)\\
 & =\sqrt{n}\frac{F_{n}\left(t\right)-F_{\lambda}\left(t\right)}{g^{'}\left(F(t)\right)}+O_{p}\left(n^{-\frac{1}{2}}\right),
\end{align*}

and the proof is completed considering $$\sqrt{n}(F_{n}\left(t\right)-F_{\lambda}\left(t\right))\xrightarrow{d}N(0,\sum_{r=1}^{m}\lambda_{r}F_{\left(r\right)}\left(t\right)(1-F_{\left(r\right)}\left(t\right)))$$
when $n^{*}\rightarrow+\infty$.
\end{proof}

\subsection{Maximum likelihood approach}

The maximum likelihood (ML) estimator of the population CDF requires
perfect ranking assumption. If we let $Y_{\left(r\right)}^{+}=\sum_{j=1}^{n_{r}}\mathbb{I}\left(Y_{\left(r\right)}^{j}\leq t\right)$,
then $Y_{\left(r\right)}^{+}$ follows a Binomial distribution with
the mass parameter $n_{r}$ and the success probability $F_{\left(r\right)}\left(t\right)=\frac{1}{r}\sum_{i=1}^{r}\mathbb{B}\left(F(t\right),i,m+1-i)$.
Therefore, using the values of $y_{\left(r\right)}^{+}$, the likelihood
function of $F\left(t\right)$ can be written as:

\[
L\left(F\left(t\right)\right)=\prod_{r=1}^{m}\binom{n_{r}}{y_{\left(r\right)}^{+}}F_{\left(r\right)}\left(t\right)^{y_{\left(r\right)}^{+}}\left(1-F_{\left(r\right)}\left(t\right)\right)^{n_{r}-y_{\left(r\right)}^{+}},
\]

Thus, the log-likelihood function is obtained as:

\begin{align*}
\mathscr{L}\left(F\left(t\right)\right) & =log\left(L\left(F\left(t\right)\right)\right)\\
 & =\sum_{r=1}^{m}\left[log\binom{n_{r}}{y_{\left(r\right)}^{+}}+y_{\left(r\right)}^{+}log\left(F_{\left(r\right)}\left(t\right)\right)+\left(n_{r}-y_{\left(r\right)}^{+}\right)log\left(1-F_{\left(r\right)}\left(t\right)\right)\right].
\end{align*}

Therefore, the ML estimator of the CDF $\left(F\left(t\right)\right)$
is defined as: 
\[
F_{ml}\left(t\right)=\arg\max_{F\left(t\right)\in\left[0,1\right]}\mathscr{L}\left(F\left(t\right)\right).
\]

To show the existence and uniqueness of $F_{ml}\left(t\right)$, we
need to prove that $\mathscr{L}\left(F\left(t\right)\right)$ is concave
in $F\left(t\right)$. This requires the log-concavity of $F_{\left(r\right)}\left(t\right)=\frac{1}{r}\sum_{i=1}^{r}\mathbb{B}\left(F(t\right),i,m+1-i)$
in $F(t)$, which follows from Theorem 2 in \citet{Mu}.

To establish the asymptotic normality of $F_{ml}\left(t\right)$,
we first need to mention the following result from \citet{Hjort}.
\begin{thm}
Suppose that $A_{n}\left(h\right)$ is a concave function in $h$
which can be represented as $A_{n}\left(h\right)=U_{n}h+\frac{1}{2}h^{2}V+r_{n}\left(h\right)$
, where $U_{n}$, $r_{n}$ are two sequences of random variables and
$V$ is a positive number. If $r_{n}\xrightarrow{p}0$ and $U_{n}\xrightarrow{d}U$
when $n$ goes to infinity, then $\alpha_{n}=\arg\max_{h}A_{n}\left(h\right)\xrightarrow{d}V^{-1}U$.
\end{thm}

Now, we are ready to present the following asymptotic result for $F_{ml}\left(t\right)$.
\begin{thm}
Let $\left\{ Y_{\left(r\right)}^{j},j=1,\ldots,n_{r};r=1,\ldots m\right\} $
be a MinPNS sample from a population with CDF $\left(F\left(t\right)\right)$
which is obtained under the perfect ranking assumption. If we let
$n_{r}\rightarrow+\infty,$ such that $n=\sum_{r=1}^{m}n_{r}\rightarrow+\infty$,
$\frac{n_{r}}{n}\rightarrow\lambda_{r}\in\left(0,1\right)$, and $\sum_{r=1}^{m}\lambda_{r}=1$,
then $\sqrt{n}\left(F_{ml}\left(t\right)-F(t)\right)$ converges in
distribution to a mean zero normal distribution with the variance
\[
\sigma_{ml}^{2}=\left(\sum_{r=1}^{m}\lambda_{r}\frac{\beta_{\left(r\right)}^{2}\left(F\left(t\right)\right)}{\mathbb{B}_{\left(r\right)}\left(F\left(t\right)\right)\left(1-\mathbb{B}_{\left(r\right)}\left(F\left(t\right)\right)\right)}\right)^{-1},
\]

where $\beta_{\left(r\right)}\left(F\left(t\right)\right)=\left(\mathbb{B}_{\left(r\right)}\left(F\left(t\right)\right)\right)^{/}=\frac{1}{r}\sum_{i=1}^{r}\mathbb{\beta}\left(F\left(t\right),i,m+1-i\right)$,
and $\mathbb{\beta}(t,i,m+1-i)$ is the probability density function
of the beta distribution with the parameters $i$ and $m+1-i$, evaluated
at the point $t$.
\end{thm}

\begin{proof}
Using Taylor series expansion, one can write:

\begin{eqnarray*}
 &  & \mathscr{L}\left(F\left(t\right)+\frac{h}{\sqrt{n}}\right)-\mathscr{L}\left(F\left(t\right)\right)\\
 &  & =\sum_{r=1}^{m}\left[y_{\left(r\right)}^{+}log\left(\frac{\mathbb{B}_{\left(r\right)}\left(F\left(t\right)+\frac{h}{\sqrt{n}}\right)}{\mathbb{B}_{\left(r\right)}\left(F\left(t\right)\right)}\right)+\left(n_{r}-y_{\left(r\right)}^{+}\right)log\left(\frac{1-\mathbb{B}_{\left(r\right)}\left(F\left(t\right)+\frac{h}{\sqrt{n}}\right)}{1-\mathbb{B}_{\left(r\right)}\left(F\left(t\right)\right)}\right)\right]\\
 &  & =\sum_{r=1}^{m}y_{\left(r\right)}^{+}log\left(1+\frac{\mathbb{B}_{\left(r\right)}\left(F\left(t\right)+\frac{h}{\sqrt{n}}\right)-\mathbb{B}_{\left(r\right)}\left(F\left(t\right)\right)}{\mathbb{B}_{\left(r\right)}\left(F\left(t\right)\right)}\right)\\
 &  & +\sum_{r=1}^{m}\left(n_{r}-y_{\left(r\right)}^{+}\right)log\left(1-\frac{\mathbb{B}_{\left(r\right)}\left(F\left(t\right)+\frac{h}{\sqrt{n}}\right)-\mathbb{B}_{\left(r\right)}\left(F\left(t\right)\right)}{1-\mathbb{B}_{\left(r\right)}\left(F\left(t\right)\right)}\right)\\
 &  & =\sum_{r=1}^{m}\left[y_{\left(r\right)}^{+}\left(\frac{\mathbb{B}_{\left(r\right)}\left(F\left(t\right)+\frac{h}{\sqrt{n}}\right)-\mathbb{B}_{\left(r\right)}\left(F\left(t\right)\right)}{\mathbb{B}_{\left(r\right)}\left(F\left(t\right)\right)}\right)-\frac{1}{2}y_{\left(r\right)}^{+}\left(\frac{\mathbb{B}_{\left(r\right)}\left(F\left(t\right)+\frac{h}{\sqrt{n}}\right)-\mathbb{B}_{\left(r\right)}\left(F\left(t\right)\right)}{\mathbb{B}_{\left(r\right)}\left(F\left(t\right)\right)}\right)^{2}\right]\\
 &  & -\sum_{r=1}^{m}\left(n_{r}-y_{\left(r\right)}^{+}\right)\left(\frac{\mathbb{B}_{\left(r\right)}\left(F\left(t\right)+\frac{h}{\sqrt{n}}\right)-\mathbb{B}_{\left(r\right)}\left(F\left(t\right)\right)}{1-\mathbb{B}_{\left(r\right)}\left(F\left(t\right)\right)}\right)\\
 &  & -\frac{1}{2}\sum_{r=1}^{m}\left(n_{r}-y_{\left(r\right)}^{+}\right)\left(\frac{\mathbb{B}_{\left(r\right)}\left(F\left(t\right)+\frac{h}{\sqrt{n}}\right)-\mathbb{B}_{\left(r\right)}\left(F\left(t\right)\right)}{1-\mathbb{B}_{\left(r\right)}\left(F\left(t\right)\right)}\right)^{2}+R\left(h,n\right)
\end{eqnarray*}

\begin{eqnarray*}
 &  & =\sum_{r=1}^{m}\left[\left(\mathbb{B}_{\left(r\right)}\left(F\left(t\right)+\frac{h}{\sqrt{n}}\right)-\mathbb{B}_{\left(r\right)}\left(F\left(t\right)\right)\right)\left(\frac{y_{\left(r\right)}^{+}}{\mathbb{B}_{\left(r\right)}\left(F\left(t\right)\right)}-\frac{n_{r}-y_{\left(r\right)}^{+}}{1-\mathbb{B}_{\left(r\right)}\left(F\left(t\right)\right)}\right)\right]\\
 &  & -\frac{1}{2}\sum_{r=1}^{m}\left[\left(\mathbb{B}_{\left(r\right)}\left(F\left(t\right)+\frac{h}{\sqrt{n}}\right)-\mathbb{B}_{\left(r\right)}\left(F\left(t\right)\right)\right)^{2}\left(\frac{y_{\left(r\right)}^{+}}{\left(\mathbb{B}_{\left(r\right)}\left(F\left(t\right)\right)\right)^{2}}-\frac{n_{r}-y_{\left(r\right)}^{+}}{\left(1-\mathbb{B}_{\left(r\right)}\left(F\left(t\right)\right)\right)^{2}}\right)\right]+R\left(h,n\right)\\
 &  & =\sum_{r=1}^{m}\left[\left(\mathbb{B}_{\left(r\right)}\left(F\left(t\right)+\frac{h}{\sqrt{n}}\right)-\mathbb{B}_{\left(r\right)}\left(F\left(t\right)\right)\right)\left(\frac{y_{\left(r\right)}^{+}-n_{r}y_{\left(r\right)}^{+}}{\mathbb{B}_{\left(r\right)}\left(F\left(t\right)\right)\left(1-\mathbb{B}_{\left(r\right)}\left(F\left(t\right)\right)\right)}\right)\right]\\
 &  & -\frac{1}{2}\sum_{r=1}^{m}\left[\left(\mathbb{B}_{\left(r\right)}\left(F\left(t\right)+\frac{h}{\sqrt{n}}\right)-\mathbb{B}_{\left(r\right)}\left(F\left(t\right)\right)\right)^{2}\left(\frac{y_{\left(r\right)}^{+}}{\left(\mathbb{B}_{\left(r\right)}\left(F\left(t\right)\right)\right)^{2}}-\frac{n_{r}-y_{\left(r\right)}^{+}}{\left(1-\mathbb{B}_{\left(r\right)}\left(F\left(t\right)\right)\right)^{2}}\right)\right]+R\left(h,n\right),
\end{eqnarray*}

where $R\left(h,n\right)$ is a remainder term. Note that since $Y_{\left(r\right)}^{+}$
follows a binomial distribution with mass parameter $n_{r}$ and success
probability $\mathbb{B}_{\left(r\right)}\left(F\left(t\right)\right)$,
one can write:

\begin{align*}
\frac{y_{\left(r\right)}^{+}-n_{r}y_{\left(r\right)}^{+}}{\sqrt{n_{r}}\left(\mathbb{B}_{\left(r\right)}\left(F\left(t\right)\right)\left(1-\mathbb{B}_{\left(r\right)}\left(F\left(t\right)\right)\right)\right)}\xrightarrow{d}Z_{r}\equiv N\left(0,\frac{1}{\mathbb{B}_{\left(r\right)}\left(F\left(t\right)\right)\left(1-\mathbb{B}_{\left(r\right)}\left(F\left(t\right)\right)\right)}\right) & when\textrm{ }n_{r}\rightarrow+\infty.\\
\end{align*}

Moreover, we have:
\begin{align*}
\sqrt{n_{r}}\left(\mathbb{B}_{\left(r\right)}\left(F\left(t\right)\right)\left(1-\mathbb{B}_{\left(r\right)}\left(F\left(t\right)\right)\right)\right)\rightarrow\sqrt{\lambda_{r}}h\mathbb{\beta}_{\left(r\right)}\left(F\left(t\right)\right) & when\textrm{ }n\rightarrow+\infty.\\
\end{align*}

Therefore,

\begin{eqnarray*}
 &  & \mathscr{L}\left(F\left(t\right)+\frac{h}{\sqrt{n}}\right)-\mathscr{L}\left(F\left(t\right)\right)\xrightarrow{}\\
 &  & h\sum_{r=1}^{m}\sqrt{\lambda_{r}}\mathbb{\beta}_{\left(r\right)}\left(F\left(t\right)\right)Z_{r}-\frac{1}{2}h^{2}\left(\sum_{r=1}^{m}\lambda_{r}\frac{\beta_{\left(r\right)}^{2}\left(F\left(t\right)\right)}{\mathbb{B}_{\left(r\right)}\left(F\left(t\right)\right)\left(1-\mathbb{B}_{\left(r\right)}\left(F\left(t\right)\right)\right)}\right)\\
 &  & =hZ-\frac{1}{2}h^{2}\left(\sigma_{ml}^{2}\right)^{-1},
\end{eqnarray*}

where $Z=\sum_{r=1}^{m}\sqrt{\lambda_{r}}\mathbb{\beta}_{\left(r\right)}\left(F\left(t\right)\right)Z_{r}$
follows a mean zero normal distribution with the variance $\left(\sigma_{ml}^{2}\right)^{-1}$.
Now, if we let $A_{n}\left(h\right)\equiv\mathscr{L}\left(F\left(t\right)+\frac{h}{\sqrt{n}}\right)-\mathscr{L}\left(F\left(t\right)\right)$,
then it can be written as $A_{n}\left(h\right)=-\frac{1}{2}h^{2}\left(\sigma_{ml}^{2}\right)^{-1}+hU_{n}+r_{n}\left(h\right)$,
where

\[
U_{n}\equiv\sum_{r=1}^{m}\left[\left(\mathbb{B}_{\left(r\right)}\left(F\left(t\right)+\frac{h}{\sqrt{n}}\right)-\mathbb{B}_{\left(r\right)}\left(F\left(t\right)\right)\right)\left(\frac{y_{\left(r\right)}^{+}-n_{r}y_{\left(r\right)}^{+}}{\mathbb{B}_{\left(r\right)}\left(F\left(t\right)\right)\left(1-\mathbb{B}_{\left(r\right)}\left(F\left(t\right)\right)\right)}\right)\right]
\]

converges in distribution to a mean zero normal distribution with
the variance $\left(\sigma_{ml}^{2}\right)^{-1}$, and

\begin{eqnarray*}
 &  & r_{n}\left(h\right)\equiv-\frac{1}{2}\left\{ \sum_{r=1}^{m}\left[\left(\mathbb{B}_{\left(r\right)}\left(F\left(t\right)+\frac{h}{\sqrt{n}}\right)-\mathbb{B}_{\left(r\right)}\left(F\left(t\right)\right)\right)^{2}\left(\frac{y_{\left(r\right)}^{+}}{\left(\mathbb{B}_{\left(r\right)}\left(F\left(t\right)\right)\right)^{2}}-\frac{n_{r}-y_{\left(r\right)}^{+}}{\left(1-\mathbb{B}_{\left(r\right)}\left(F\left(t\right)\right)\right)^{2}}\right)\right]\right\} \\
 &  & -\frac{h^{2}}{\sigma_{ml}^{2}}+R\left(h,n\right)
\end{eqnarray*}

converges in probability to zero when $n\rightarrow+\infty$. Thus,
it follows from Theorem 2 that $\alpha_{n}=\arg\max_{h}A_{n}\left(h\right)\xrightarrow{d}N\left(0,\sigma_{ml}^{2}\right)$.
Note that since $\mathscr{L}\left(F\left(t\right)\right)$ is maximized
at $F_{ml}\left(t\right)$, $\alpha_{n}=\sqrt{n}\left(F_{ml}\left(t\right)-F\left(t\right)\right)$
which completes the proof.
\end{proof}

\section{Comparisons\label{sec:Comparisons}}

In this section, the performance of two CDF estimators in MinPNS is
compared with those of that counterparts in SRS. Let $X_{1},\ldots,X_{n}$
be a SRS sample with the size of $n$ from a population with the CDF
of $F$. It is well-known that in SRS, the ML and MB estimators of
the CDF are the same and are obtained by:

\[
F_{n}\left(t\right)=\frac{1}{n}\sum_{i=1}^{n}\mathbb{I}\left(X_{i}\leq t\right).
\]

This estimator is unbiased and $\sqrt{n}\left(F_{n}\left(t\right)-F\left(t\right)\right)$
converges in distribution to a mean zero normal distribution with
variance $\sigma_{srs}^{2}\left(t\right)=F\left(t\right)\left(1-F\left(t\right)\right)$.

\subsection{Asymptotic Comparison}

In this subsection, the asymptotic performance of $F_{ml}\left(t\right)$
and $F_{mb}\left(t\right)$ are first compared. Then, their performances
are compared with that of $F_{n}\left(t\right)$. This requires the
perfect ranking assumption. The next result shows that under the perfect
ranking assumption, $F_{ml}\left(t\right)$ is at least as asymptotically
efficient as $F_{mb}\left(t\right)$.
\begin{thm}
Under the perfect ranking assumption, $\sigma_{ml}^{2}\left(t\right)\leq\sigma_{ml}^{2}\left(t\right)$,
and equality holds if and only if either $m=1,$ or $\exists r\in\left\{ 1,\ldots,m\right\} $
such that $\lambda_{r}=1$.
\end{thm}

\begin{proof}
Let $R$ be a discrete random variable with the support set $\left\{ 1,\ldots,m\right\} ,$
such that

\[
\mathbb{P}\left(R=r\right)=\frac{\lambda_{r}\beta_{\left(r\right)}\left(F\left(t\right)\right)}{\sum_{r=1}^{m}\lambda_{r}\beta_{\left(r\right)}\left(F\left(t\right)\right)},
\]

for $r\in\left\{ 1,\ldots,m\right\} $. Now, we define the vector
$\mathbf{C}=\left\{ C_{1},\ldots,C_{m}\right\} $ such that $C_{r}=\frac{\beta_{\left(r\right)}\left(F\left(t\right)\right)}{\mathbb{B}_{\left(r\right)}\left(F\left(t\right)\right)\left(1-\mathbb{B}_{\left(r\right)}\left(F\left(t\right)\right)\right)},$
for $r\in\left\{ 1,\ldots,m\right\} $. Using Jensen's inequality,
one can write: 
\[
\left(\mathbb{E}\left(C_{R}\right)\right)^{-1}\leq\mathbb{E}\left(\frac{1}{C_{R}}\right),
\]

and therefore:

\begin{align*}
\left(\sum_{r=1}^{m}\lambda_{r}\frac{\beta_{\left(r\right)}^{2}\left(F\left(t\right)\right)}{\mathbb{B}_{\left(r\right)}\left(F\left(t\right)\right)\left(1-\mathbb{B}_{\left(r\right)}\left(F\left(t\right)\right)\right)}\right)^{-1}\times\left(\sum_{r=1}^{m}\lambda_{r}\beta_{\left(r\right)}\left(F\left(t\right)\right)\right)\\
\leq\left(\frac{\sum_{r=1}^{m}\lambda_{r}\mathbb{B}_{\left(r\right)}\left(F\left(t\right)\right)\left(1-\mathbb{B}_{\left(r\right)}\left(F\left(t\right)\right)\right)}{\sum_{r=1}^{m}\lambda_{r}\beta_{\left(r\right)}\left(F\left(t\right)\right)}\right).
\end{align*}

Thus, $\sigma_{ml}^{2}\left(t\right)\leq\sigma_{mb}^{2}\left(t\right)$.
The equality holds if $C_{R}$ has a degenerated distribution and
this happens if and only if either $m=1,$ or $\exists r\in\left\{ 1,\ldots,m\right\} $
such that $\lambda_{r}=1$.

In order to have a better understanding of the asymptotic performances
of the estimators, we have compared their asymptotic efficiencies.
To do so, we set $m\in\left\{ 3,5\right\} $ and consider three different
scenarios for the vector $\mathbf{\lambda}=\left(\lambda_{1},\ldots,\lambda_{m}\right)$
as shown in Table \ref{table2}.
\end{proof}
\begin{table}
\centering %
\begin{tabular}{cccc}
\hline 
Scenario & $\mathbf{\lambda}=\left(\lambda_{1},\ldots,\lambda_{m}\right)$ & $m=3$ & $m=5$\tabularnewline
\hline 
A & $\lambda_{A}$ & $\left(\frac{4}{6},\frac{1}{6},\frac{1}{6}\right)$ & $\left(\frac{4}{10},\frac{2}{10},\frac{2}{10},\frac{1}{10},\frac{1}{10}\right)$\tabularnewline
B & $\lambda_{B}$ & $\left(\frac{2}{6},\frac{2}{6},\frac{2}{6}\right)$ & $\left(\frac{2}{10},\frac{2}{10},\frac{2}{10},\frac{2}{10},\frac{2}{10}\right)$\tabularnewline
C & $\lambda_{C}$ & $\left(\frac{1}{6},\frac{1}{6},\frac{4}{6}\right)$ & $\left(\frac{1}{10},\frac{1}{10},\frac{2}{10},\frac{2}{10},\frac{4}{10}\right)$\tabularnewline
\hline 
\end{tabular}\caption{\label{table2}{\small{}Three different scenarios for the vector}
$\mathbf{\lambda}=\left(\lambda_{1},\ldots,\lambda_{m}\right)$}
\end{table}

It is clear from Table \ref{table2} that the confidence of the ranker
to find the sample unit with the lowest rank in a set of size $m$
decreases as we move from scenario $A$ to scenario $C$. In order
to compare the asymptotic performances of $F_{mb}\left(t\right)$,
$F_{ml}\left(t\right)$ with that of $F_{n}\left(t\right)$, we have
defined the asymptotic relative efficiency (ARE) of $F_{z}\left(t\right)$
with respect to $F_{n}\left(t\right)$ as:

\[
ARE\left(t\right)=\frac{\sigma_{srs}^{2}\left(t\right)}{\sigma_{z}^{2}\left(t\right)},
\]

for $z=mb,$ and $ml$. With this definition, an $ARE\left(t\right)$
larger than one indicates that $F_{z}\left(t\right)$ is asymptotically
more efficient than $F_{n}\left(t\right)$ at point $t$.

\begin{landscape}

\begin{figure}
\centering\includegraphics[width=20cm,height=15cm]{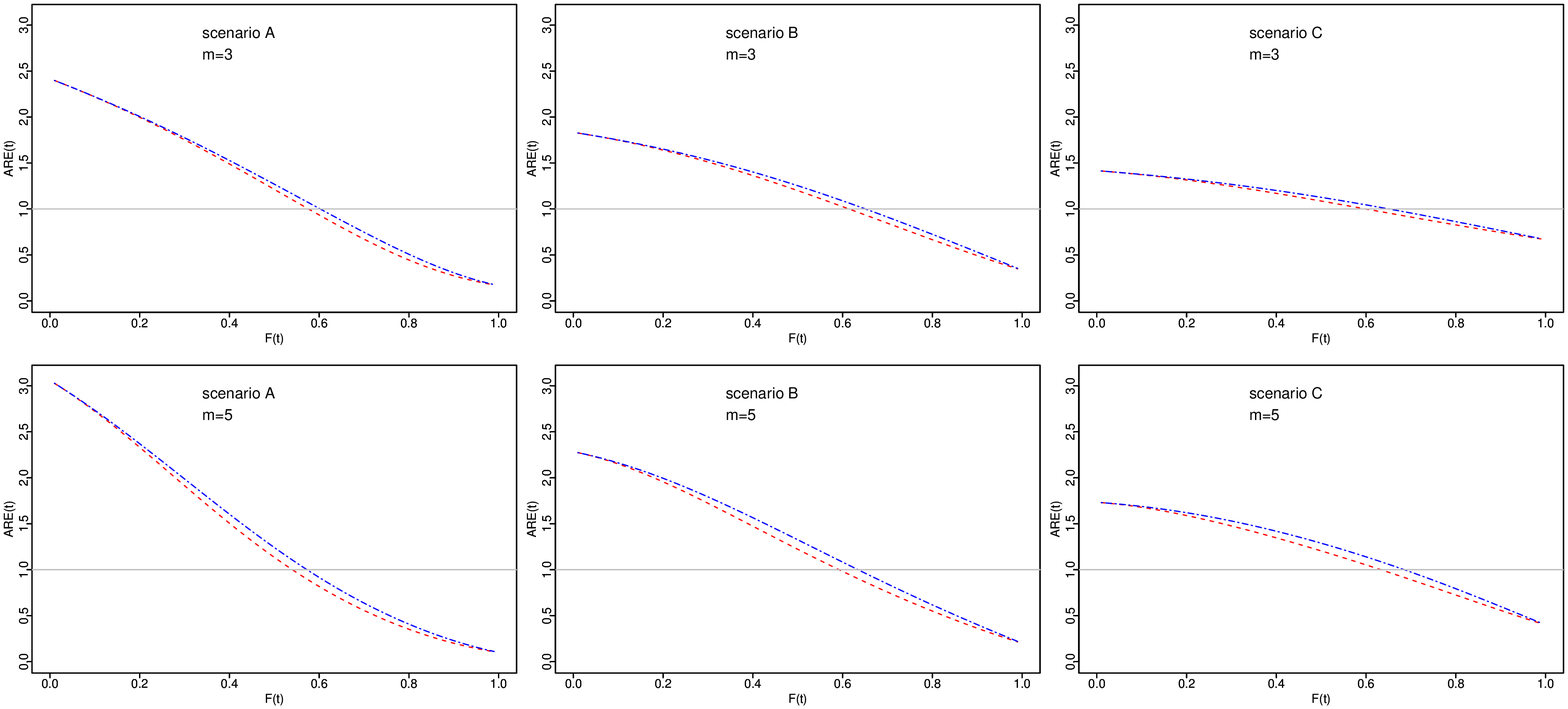}

\caption{The asymptotic relative efficiency (ARE) of $F_{mb}\left(t\right)$
with respect to $F_{n}\left(t\right)$ is represented by a red dashed
line, while that of $F_{ml}\left(t\right)$ with respect to $F_{n}\left(t\right)$
is represented by a blue two-dashed line for $m\in\left\{ 3,5\right\} $
under three different scenarios for the ties.}

\label{Fig1}
\end{figure}

\end{landscape}

Figure \ref{Fig1} shows that the $ARE$s as a function of $F\left(t\right)$
under assumption of perfect ranking. It is worth mentioning that under
this assumption, $ARE$s are independent of $F$. It can be observed
from Figure \ref{Fig1} that a sizable asymptotic efficiency gain
is obtained at the lower tail of the parent distribution as a result
of using $F_{mb}\left(t\right)$ or $F_{ml}\left(t\right)$ instead
of $F_{n}\left(t\right)$. As it is anticipated, $ARE\left(t\right)$
increases with the set size ($m$) at the lower tail of the parent
distribution. However, it decreases when we move from scenario $A$
to scenario $C$. A small asymptotic efficiency gain is observed when
$F_{ml}\left(t\right)$ is used instead of $F_{mb}\left(t\right)$
for all values of $t$.

\subsection{Finite sample size comparisons}

In this sub-section, the finite sample size performances of the CDF
estimators are evaluated in the MinPNS and SRS designs using the Monte
Carlo simulation. Let $t=Q_{p}$, for $p\in\left[0,1\right],$where
$Q_{p}$ is the $p$th quantile of the population distribution. Since
$F_{n}\left(t\right)$ is an unbiased estimator for the CDF of the
population $\left(F\left(t\right)\right)$, the relative efficiency
(RE) of $F_{mb}\left(t\right)$ and $F_{ml}\left(t\right)$ is defined
as the ratio of the variance of $F_{n}\left(t\right)$ to the mean
square error (MSE) of $F_{mb}\left(t\right)$ and $F_{ml}\left(t\right)$,
respectively, i.e.

\[
RE\left(p\right)=\frac{\mathbb{V}\left(F_{n}\left(Q_{p}\right)\right)}{MSE\left(F_{z}\left(Q_{p}\right)\right)},
\]

for $z=mb$ and $ml.$ With this definition, an $RE\left(p\right)>1$
means that $F_{z}\left(t\right)$ is more efficient than $F_{n}\left(t\right)$
at the point $t=Q_{p}$. By following lines of Proposition 3 in \citet{Wang2012},
one can simply show that, under perfect ranking assumption, $RE\left(p\right)$,
as a function of $p$, does not depend on the parent distribution.

The Monte Carlo simulation was performed to estimate $RE\left(p\right)$
for both perfect and imperfect ranking cases. The ranking process
was done using the linear ranking model introduced by \citet{Dell}.
In this model, it is assumed that the variable of interest is $Y$,
however, the ranking process is done using the perceived value of
$X$. The following relation exists between $X$ and $Y$ :

\[
X=\rho\left(\frac{Y-\mu_{y}}{\sigma_{y}}\right)+\sqrt{1-\rho^{2}}Z,
\]

where $\mu_{y}$, and $\sigma_{y}$ are the mean and standard deviation
of the random variable $Y$, respectively, the random variable $Z$
is independent from $Y$ and follows a standard normal distribution,
and parameter $\rho$ is the correlation coefficient between $X$
and $Y$ and controls the quality of ranking. Setting $\rho=1$ gives
perfect ranking cases, setting $\rho=0$ gives completely random ranking
cases, and choosing $\rho\in\left(0,1\right)$ provides a ranking
which is not perfect but is better than random. We set $n\in\left\{ 30,60,120,600\right\} $,
$m\in\left\{ 3,5\right\} $ and $\rho\in\left\{ 1,0.75\right\} $.
Vector $\mathbf{n}=\left(n_{1},\ldots,n_{m}\right)$ is obtained under
three different scenarios similar to those in the Table \ref{table2}
by multiplying the sample size $n$ by the vector $\mathbf{\lambda}=\left(\lambda_{1},\ldots,\lambda_{m}\right)$,
i.e. $\mathbf{n}=\left(\lambda_{1}n,\ldots,\lambda_{m}n\right)$.
For each combination of $\left(\mathbf{n},m,\rho\right),$ we drew
$100,000$ random samples from the MinPNS and SRS designs when the
parent distribution was standard uniform and standard exponential.
Finally, $RE\left(p\right)$ was estimated using $100,000$ random
samples for $p\in\left\{ 0.01,\ldots,0.99\right\} $.

Here, only the simulation results for the standard uniform distribution
are reported in Figures \ref{Fig2}-\ref{Fig5}. This is because it
seems that the type of parent distribution does not have much effect
on the pattern of $RE$s in the imperfect ranking case, as well.

\begin{landscape}

\begin{figure}
\centering\includegraphics[width=20cm,height=15cm]{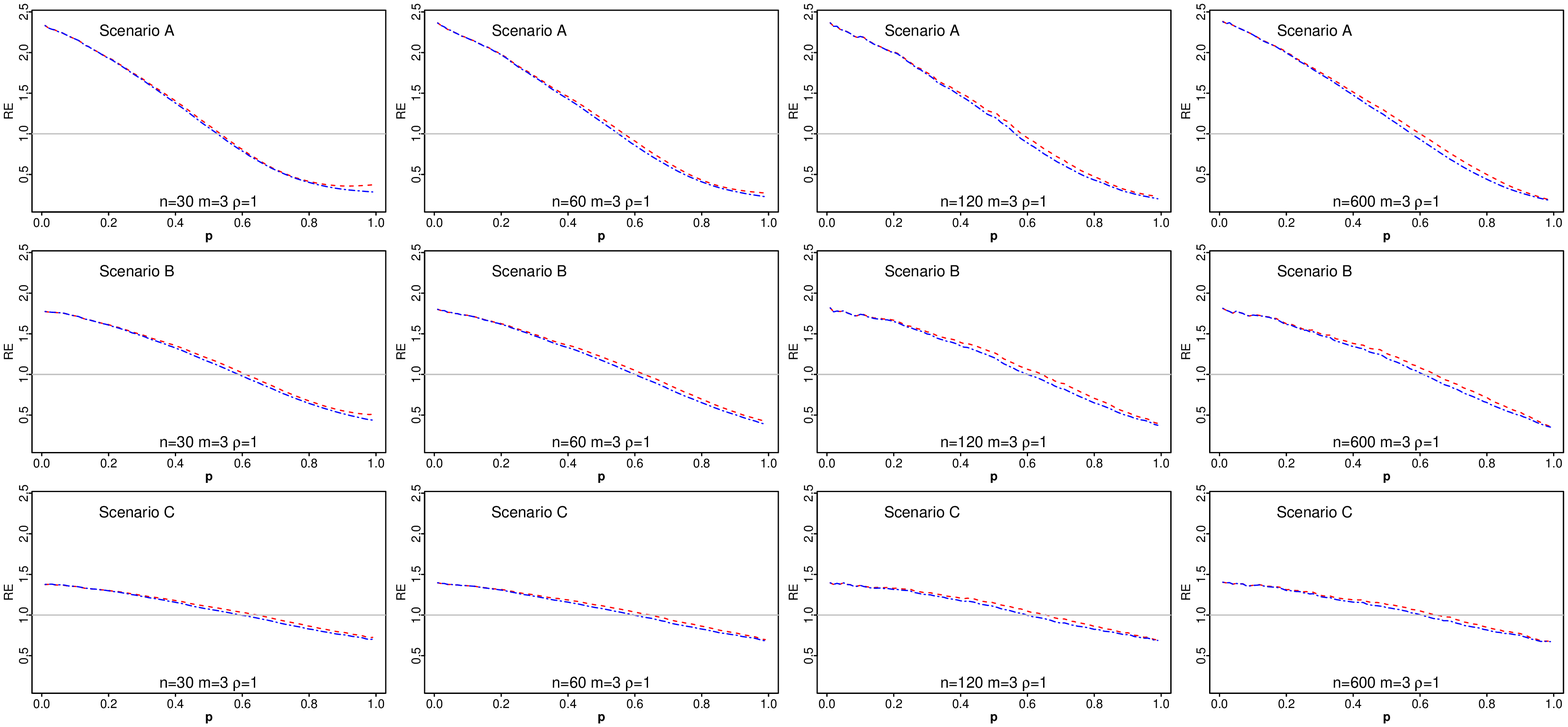}

\caption{The relative efficiency (RE) of $F_{mb}\left(t\right)$ with respect
to $F_{n}\left(t\right)$ represented by a red dashed line and the
RE of $F_{ml}\left(t\right)$ with respect to $F_{n}\left(t\right)$
represented by a blue two-dashed line for $n\in\left\{ 30,60,120,600\right\} $,
$m=3$, and $\rho=1$ under three different scenarios for the ties
when the parent distribution is standard uniform.}

\label{Fig2}
\end{figure}

\begin{figure}
\centering\includegraphics[width=20cm,height=15cm]{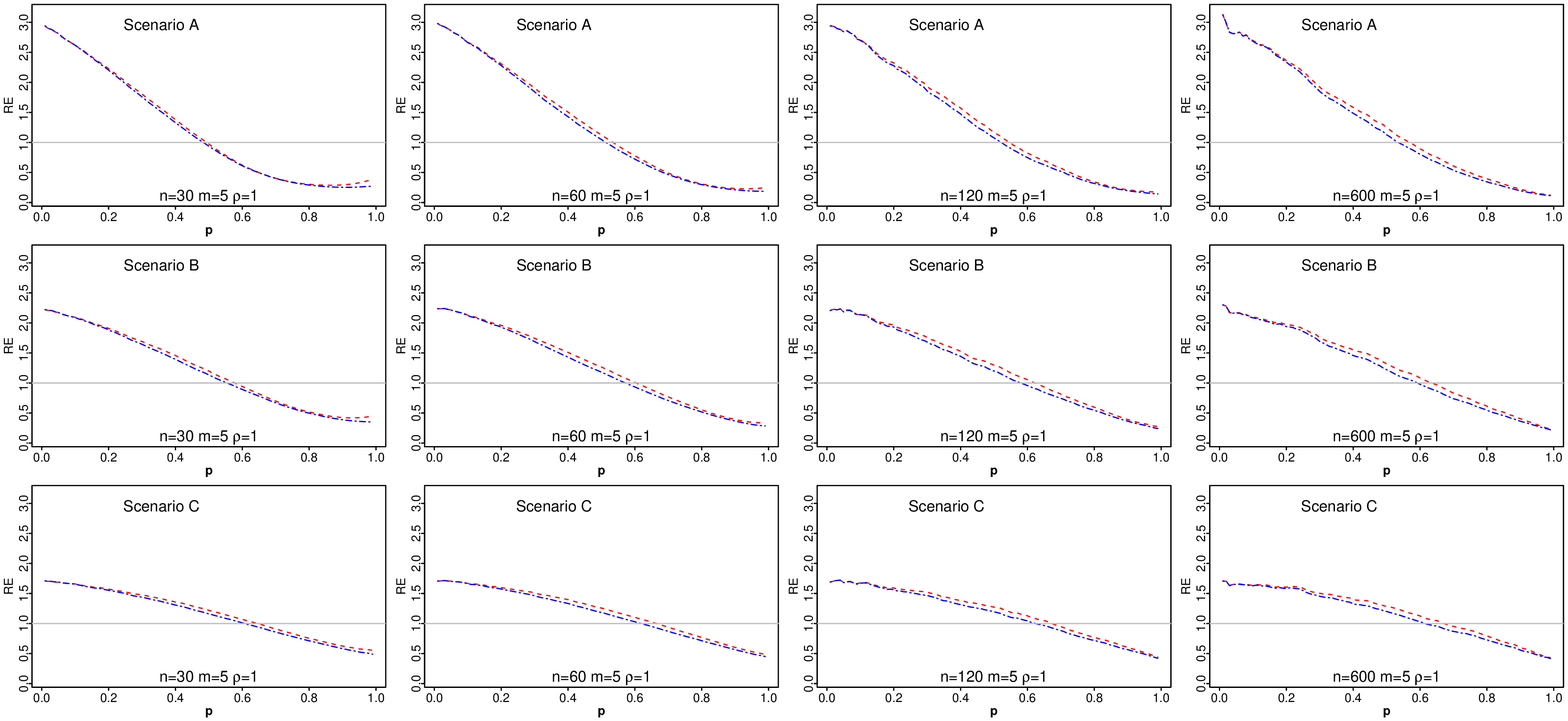}

\caption{The relative efficiency (RE) of $F_{mb}\left(t\right)$ with respect
to $F_{n}\left(t\right)$ represented by a red and $F_{ml}\left(t\right)$
with respect to $F_{n}\left(t\right)$ represented by a blue two-dashed
line for $n\in\left\{ 30,60,120,600\right\} $, $m=5$, and $\rho=1$
under three different scenarios for the ties when the parent distribution
is standard uniform.}

\label{Fig3}
\end{figure}

\begin{figure}
\centering\includegraphics[width=20cm,height=15cm]{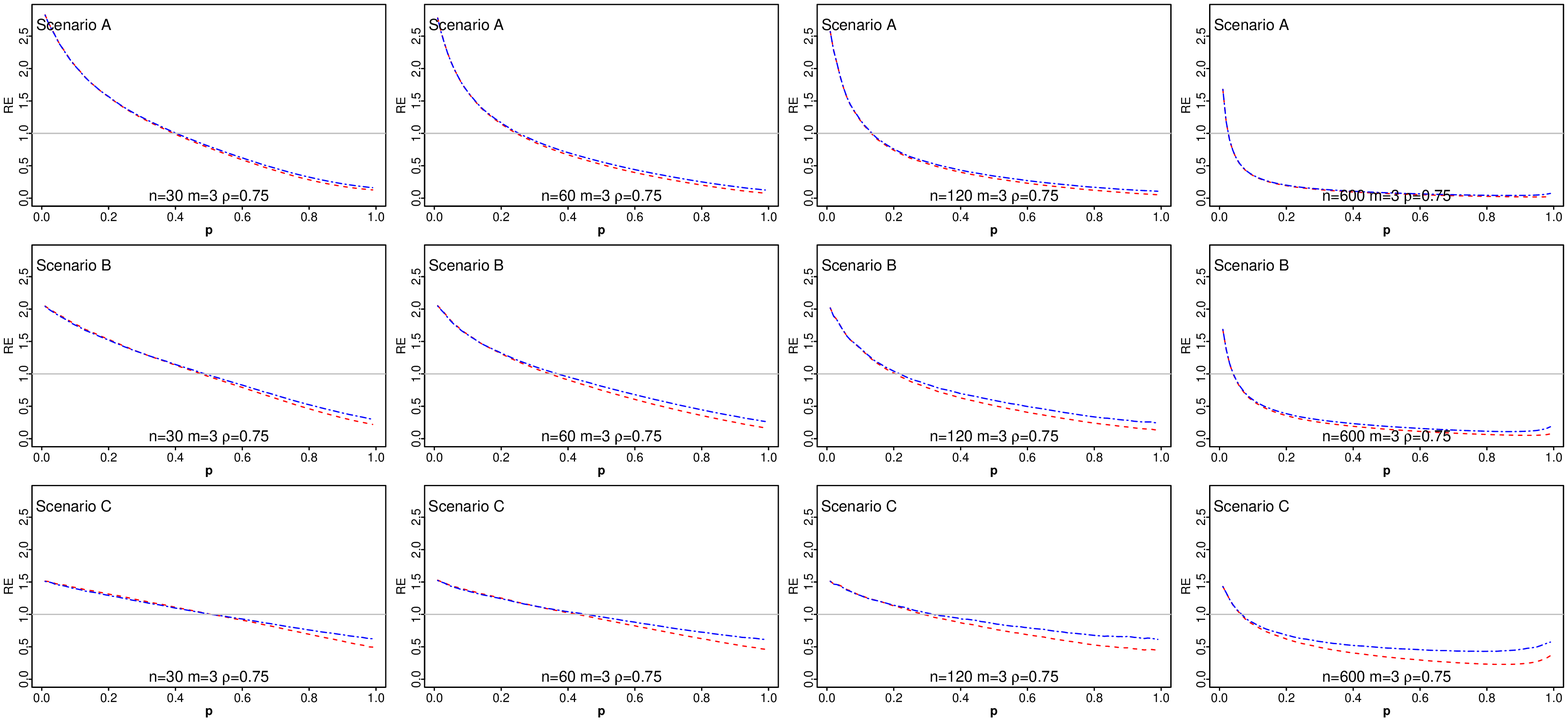}

\caption{The relative efficiency (RE) of $F_{mb}\left(t\right)$ with respect
to $F_{n}\left(t\right)$ represented by a red dashed line and $F_{ml}\left(t\right)$
with respect to $F_{n}\left(t\right)$ represented by a blue two-dashed
line for $n\in\left\{ 30,60,120,600\right\} $, $m=3$, and $\rho=0.75$
under three different scenarios for the ties when the parent distribution
is standard uniform.}

\label{Fig4}
\end{figure}

\begin{figure}
\centering\includegraphics[width=20cm,height=15cm]{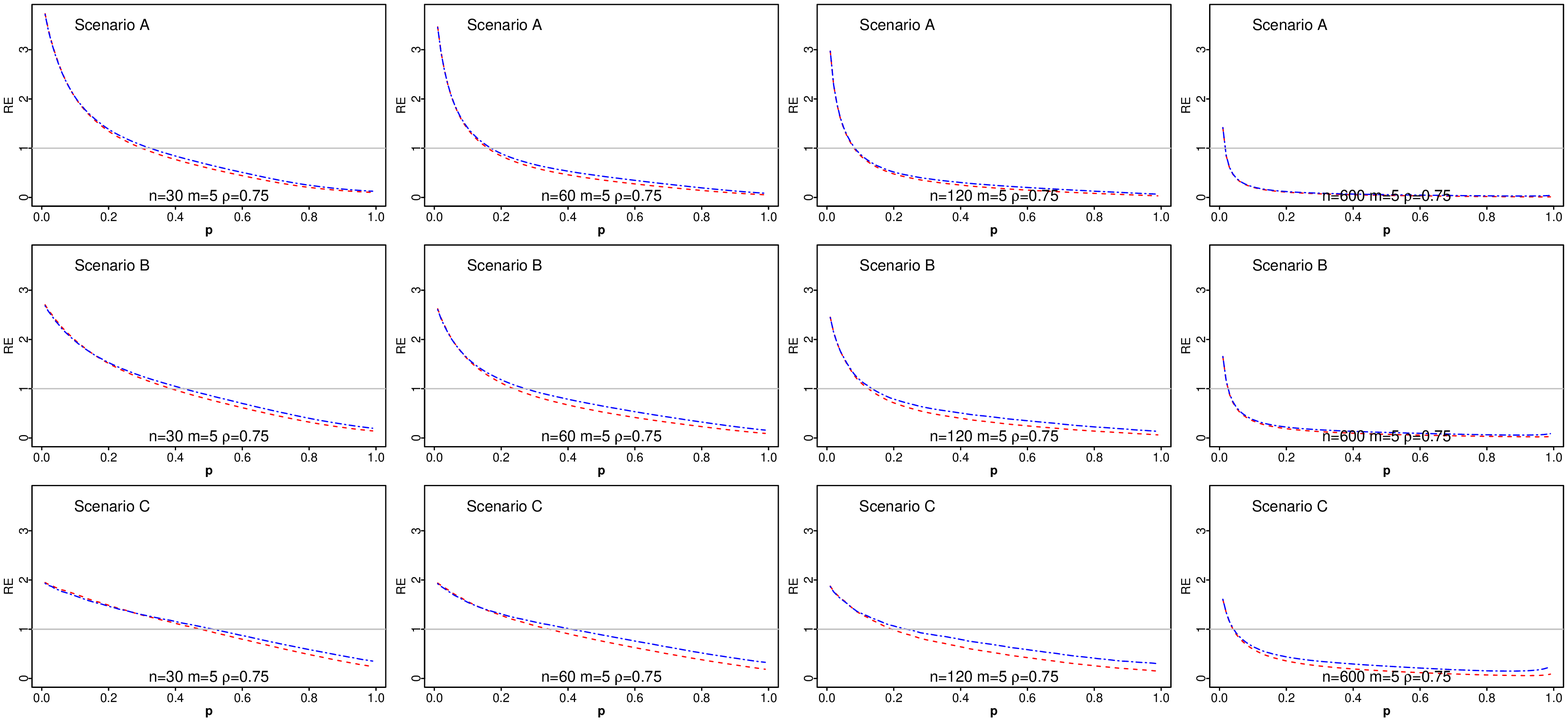}

\caption{The relative efficiency (RE) of $F_{mb}\left(t\right)$ with respect
to $F_{n}\left(t\right)$ represented by a red dashed line and $F_{ml}\left(t\right)$
with respect to $F_{n}\left(t\right)$ represented by a blue two-dashed
line for $n\in\left\{ 30,60,120,600\right\} $, $m=5$, and $\rho=0.75$
under three different scenarios for the ties when the parent distribution
is standard uniform.}

\label{Fig5}
\end{figure}

\end{landscape}

Figure \ref{Fig2} presents the simulation results for $n\in\left\{ 30,60,120,600\right\} $,
$m=3$, and the perfect ranking case ($\rho=1)$. It can be observed
from this figure that the $RE$ patterns of different estimators remain
almost the same when the sample size $n$ goes from $n=30$ to $n=600$
However, they decrease when we move from scenario $A$ to scenario
$C$ while the other parameters are kept fixed. The highest efficiency
gain is attained at the lower tail of the parent distribution under
scenario $A$ and the CDF estimators in MinPNS are around $250\%$
more efficient than their counterpart in SRS. As one intuitively expects,
the $RE$s fall below one after the median of the parent distribution.
This is not a big concern as the MinPNS scheme is designed to deal
with situations in which the researcher is interested in drawing a
statistical inference for the lower tail of the parent distribution.
Similar to what we have observed in Figure \ref{Fig1}, although the
CDF estimator based on the ML approach is more efficient than the
one based on the MB approach, the difference between their $RE$s
becomes indistinguishable in some cases, especially at the lower tail
of the parent distribution.

The simulation results for $n\in\left\{ 30,60,120,600\right\} $,
$m=5$, and $\rho=1$ are presented in Figure \ref{Fig3}. Comparing
the results of Figure \ref{Fig2} and Figure \ref{Fig3}, it can be
observed that the $RE$ patterns for $m=3$ are very similar to those
for $m=5$ with the clear difference that the $RE$ pattern are higher
for $m=5$ when they are larger than one. In Figure \ref{Fig3} ($m=5$),
the highest efficiency gain is obtained at the lower tail of the parent
distribution under scenario $A$. The CDF estimators in MinPNS are
around $300\%$ more efficient than their counterparts in SRS which
is about $50\%$ larger than the case for $m=3$ under scenario $A$
in Figure \ref{Fig2}.

The simulation results for the imperfect ranking case ($\rho=0.75$)
are depicted in Figures \ref{Fig4} and\ref{Fig5} for $m=3$ and
$m=5$, respectively. The performance patterns of the estimators in
these cases are almost the same as their counterparts in the perfect
ranking case ($\rho=1$) with the clear difference that the span of
the interval in which the $RE$s are larger than one becomes narrower
as we move from $\rho=1$ to $\rho=0.75$.

\section{Data Analysis using PNS\label{sec:Data-Analysis-using}}

Osteoporosis, which literally means porous bone, is a bone disorder
in which the density and quality of bone are reduced. As a consequence
of this disease, the bone becomes more porous and fragile and thus
becomes more likely to break. Bone tissues develop and strengthen
from the moment of birth until around the 20s when the bone is at
its densest. Bone density remains almost constant for around 10 years
later. In their late 30s, people start to lose their bone density
slowly as part of the ageing process.

Osteoporosis is a silent disease as it does not have any clear symptoms
and develops slowly and progressively. A person may not know that
he/she is suffering from it until a minor incidence such as a fall
leaves him/her with a fracture. In the absence of a low impact fracture
(a fracture occurring spontaneously or from a fall no greater than
standing height), osteoporosis is diagnosed based on a low BMD value
obtained using DXA technology. Specifically, a person is considered
to be suffering from osteoporosis if his/her BMD value obtained using
DXA technology is no larger than $0.56$ at either the femoral neck
or the lumbar spine.

Osteoporotic fractures impose a huge amount of social and economic
burden on the society. For instance, in Europe, the disability rate
due to osteoporosis is larger than that caused by cancers (with the
exception of lung cancer) and is at least comparable to that caused
by different types of chronic diseases such as asthma, rheumatoid
arthritis, and high blood pressure-related heart disease \citep{Johnel}.
Osteoporosis is a common disease and it is estimated that more than
8.9 million osteoporotic fractures happen annually worldwide resulting,
on average, in one osteoporotic fracture every 3 seconds \citep{Johnel}.
Since osteoporosis is an aging-associated disease, it is anticipated
that more and more osteoporotic fractures will observed in the future
as the global life expectancy has an increasing trend. For instance,
it is estimated that by 2035, the annual numbers and costs of osteoporosis-related
fractures will be doubled as compared to what was observed in 2010
\citep{Si}. Therefore, to minimize the impact of osteoporosis-related
fractures on the health of the growing population and the healthcare
budget, it is vital for both the government and public health officials
to monitor the prevalence of osteoporosis in the society annually.

PNS can be used in the BMD analysis for drawing statistical inference
about the lower tail of the underlying distribution as a more cost-efficient
technique than SRS. The DXA technology is expensive and not easily
accessible in some developing countries. However, a medic can simply
rank the patients based on the possibility of suffering from osteoporosis
according to his/her personal judgement or an inexpensive covariate.
The purpose of this section is to use a real dataset to show the efficiency
and applicability of the proposed procedure in a real situation where
the ranking procedure is performed using an ordinal concomitant variable
so that the ties happen naturally.

\subsection{Data description}

The World Health Organization (WHO) has recommended using 20-29 year-old
non-Hispanic white females from the dataset of the Third National
Health and Nutrition Examination Survey (NHANES III) conducted by
the National Center for Health Statistics (NCHS) and Centers for Disease
Control and Prevention to evaluate the nutritional status and health
of a representative sample of the non-institutionalized civilian US
population as the reference group for the calculation of the BMD T-score
at the femoral neck for both men and women. This national survey which
was conducted on 33994 American people from 1988 to 1994 is available
online at \href{\%20http://www.cdc.gov/nchs/nhanes/nh3data.htm}{ http://www.cdc.gov/nchs/nhanes/nh3data.htm}\footnote{Access date: Sep 2020}.
The publicty available dataset generated as a result of this survey
includes several measurements related to the nutritional status and
health of the people included in the survey. Therefore, it can be
used to show how PNS can be efficiently used to estimate the prevalence
of osteoporosis. Osteoporosis is more common among women than men.
Worldwide, 1 in 3 women aged 50 and over will experience at least
one osteoporotic fracture, as will 1 in 5 men aged 50 and over \citep{Melton1992,Melton1998,Kanis}.
Due to the importance of the prevalence of osteoporosis among elderly
women, we consider the BMD data of the women aged 50 and over in the
NHANES III survey as our ``\textit{hypothetical population}'' with
the size of $N=3978$. We will refer to it hereafter as the \textquoteleft BMD
dataset\textquoteright . The summary statistics of the population
are reported in Table \ref{table3}. It is observed that 667 out of
3978 BMD measurements are missing and the average BMD value for the
remaining people in the underlying population is $0.799$. Figure
\ref{Fig6} shows the histogram of the BMD data superimposed by the
normal density function. It is clear from this figure that the normal
distribution fits well to the BMD of women aged 50 and over.

\begin{table}
\centering %
\begin{tabular}{ccccccccc}
\hline 
$N$ & $\#$ NA & $Min$ & $Q_{1}$ & $Median$ & $Q_{3}$ & $Max$ & $Mean$ & $Variance$\tabularnewline
\hline 
3978 & 667 & 0.274 & 0.685 & 0.793 & 0.908 & 1.446 & 0.799 & 0.026\tabularnewline
\hline 
\end{tabular}\caption{\label{table3}The summary statistics for the BMD measurements of
women aged 50 and over}
\end{table}

\begin{figure}
\centering\includegraphics[width=15cm,height=20cm,keepaspectratio]{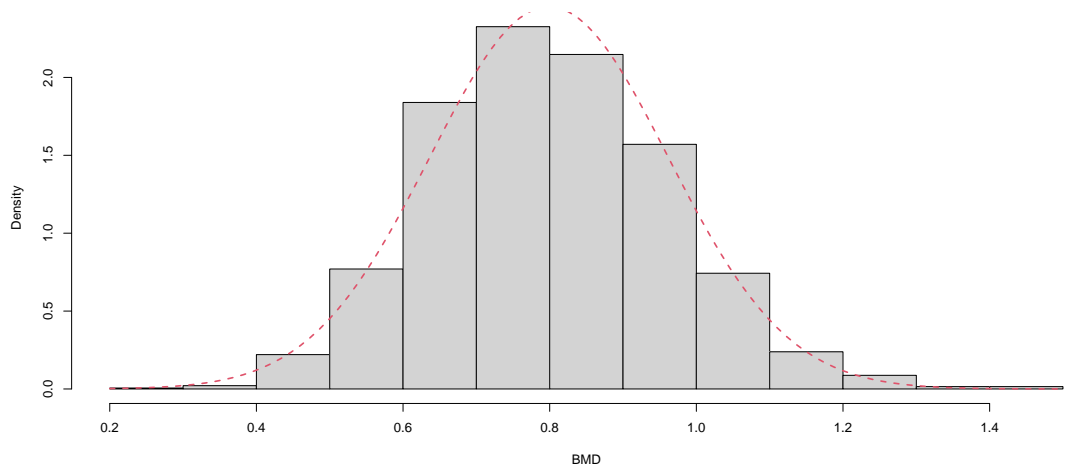}

\caption{The histogram of the bone mineral density (BMD) measurements of women
aged 50 and over along with the fitted normal density curve}

\label{Fig6}
\end{figure}

Two easily available concomitant variables for ranking are the body
mass index category (BMIC) and age decade (AD) variable. The BMIC
classifies the participants into four categories according to the
recommendation of the WHO as underweight (BMI$<18.5$), normal ($18.5\leq$BMI$<25$),
overweight ($25\leq$BMI$<30$), and obese (BMI$\geq30$). The Spearman
correlation coefficient between the BMD and BMIC variables is $0.475$.
The ordinal variable AD is obtained from the age of the participants
ranging from 50-59, 60-69, ...,90+. Moreover, the Spearman correlation
coefficient between BMD and AD variables is $-0.466.$ Since the correlation
coefficient between BMD and AD is negative, the ranking process is
done in reverse. Figure \ref{Fig7} shows the bar charts of the BMIC
and AD variables. It is clear from this figure that the frequencies
of some levels are higher than those of the others. Hence, when the
researcher wants to select the person with highest judgment probability
of suffering from osteoporosis using either BMIC or AD variables in
a set of small size, ranking ties will easily form.

\begin{figure}
\centering\includegraphics[width=15cm,height=20cm,keepaspectratio]{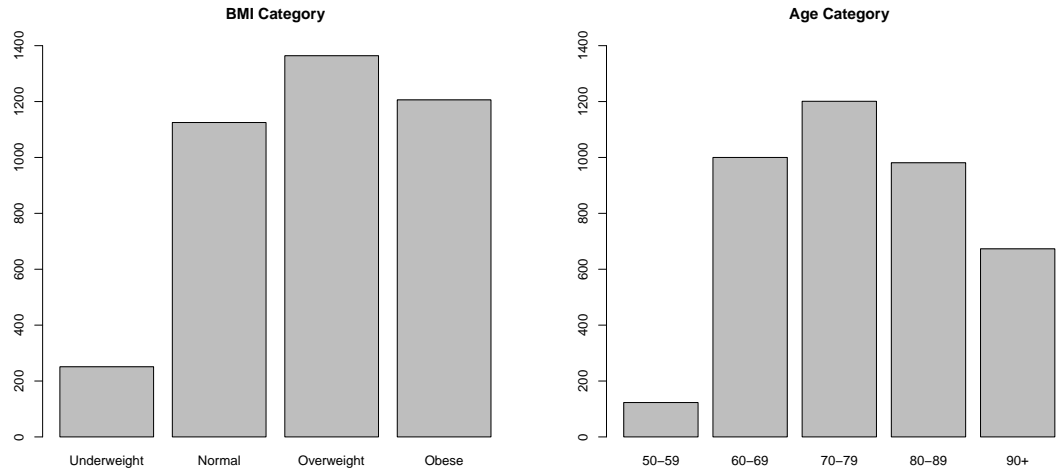}

\caption{The bar charts of the ranking variables in the BMD dataset: the BMI
category and age category}

\label{Fig7}
\end{figure}

\subsection{Empirical study}

Let the random variable $Y$ present the actual BMD value of a randomly
selected subject from the BMD dataset with the CDF $F$. The prevalence
of osteoporosis is this population is given by:

\[
F\left(0.56\right)=\mathbb{P}\left(Y\leq0.56\right)=\int_{0}^{0.56}f\left(y\right)dy=0.061,
\]

where $f\left(y\right)$ is the probability density function of the
random variable $Y$.

Here, the BMD dataset is used to compare the performances of different
CDF estimators in the MinPNS design in estimating the prevalence of
osteoporosis by employing the real ties obtained as a result of ranking
using the ordinal concomitant variables BMIC and AC. To do so, we
set $n\in\left\{ 10,20,30,40,50,100\right\} $ and $m\in\left\{ 3,4,5\right\} $.
For each combination of $\left(n,m\right)$, we draw $100,000$ random
samples from the SRS and MinPNS designs in which all the samplings
are done with replacement. To draw a MinPNS sample, we use BMIC (AC)
as the ranking variable. In each set with the size of $m,$ a person
with the lowest (highest) BMIC (AC) value is selected for actual quantification.
If there are two or more people wth the lowest (highest) BMIC (AC)
value, one of them is selected at random and the tie information is
recorded.

Tables \ref{table4} and \ref{table5} present the estimated bias,
variance, and relative efficiency (RE) of different CDF estimators
in MinPNS at point $t=0.56$ for estimating the prevalence of osteoporosis
when the ranking variable is BMIC or AC, respectively.

\begin{table}
\centering %
\begin{tabular}{cccccccccc}
\hline 
Sample Size & Set size &  & \multicolumn{3}{c}{MB estimator} &  & \multicolumn{3}{c}{ML estimator}\tabularnewline
\hline 
$n$ & $m$ &  & Bias & Variance & RE &  & Bias & Variance & RE\tabularnewline
 &  &  &  &  &  &  &  &  & \tabularnewline
 & 3 &  & -0.0101 & 0.0022 & 2.4748 &  & -0.0102 & 0.0022 & 2.4766\tabularnewline
10 & 4 &  & -0.0140 & 0.0016 & 3.0624 &  & -0.0141 & 0.0016 & 3.0732\tabularnewline
 & 5 &  & -0.0173 & 0.0013 & 3.5142 &  & -0.0174 & 0.0013 & 3.5317\tabularnewline
 &  &  &  &  &  &  &  &  & \tabularnewline
 & 3 &  & -0.0109 & 0.0010 & 2.4376 &  & -0.0109 & 0.0010 & 2.4395\tabularnewline
20 & 4 &  & -0.0150 & 0.0007 & 2.7970 &  & -0.0151 & 0.0007 & 2.8027\tabularnewline
 & 5 &  & -0.0182 & 0.0006 & 2.9775 &  & -0.0183 & 0.0006 & 2.9817\tabularnewline
 &  &  &  &  &  &  &  &  & \tabularnewline
 & 3 &  & -0.0112 & 0.0006 & 2.3220 &  & -0.0112 & 0.0006 & 2.3234\tabularnewline
30 & 4 &  & -0.0154 & 0.0005 & 2.5323 &  & -0.0154 & 0.0005 & 2.5343\tabularnewline
 & 5 &  & -0.0186 & 0.0004 & 2.5196 &  & -0.0187 & 0.0004 & 2.5198\tabularnewline
 &  &  &  &  &  &  &  &  & \tabularnewline
 & 3 &  & -0.0112 & 0.0005 & 2.2041 &  & -0.0113 & 0.0005 & 2.2054\tabularnewline
40 & 4 &  & -0.0157 & 0.0003 & 2.2563 &  & -0.0158 & 0.0003 & 2.2567\tabularnewline
 & 5 &  & -0.0187 & 0.0003 & 2.1816 &  & -0.0188 & 0.0003 & 2.1785\tabularnewline
 &  &  &  &  &  &  &  &  & \tabularnewline
 & 3 &  & -0.0113 & 0.0004 & 2.0982 &  & -0.0114 & 0.0004 & 2.0995\tabularnewline
50 & 4 &  & -0.0156 & 0.0003 & 2.0800 &  & -0.0157 & 0.0003 & 2.0787\tabularnewline
 & 5 &  & -0.0187 & 0.0002 & 1.9210 &  & -0.0188 & 0.0002 & 1.9158\tabularnewline
 &  &  &  &  &  &  &  &  & \tabularnewline
 & 3 &  & -0.0117 & 0.0002 & 1.6809 &  & -0.0117 & 0.0002 & 1.6801\tabularnewline
100 & 4 &  & -0.0157 & 0.0001 & 1.4268 &  & -0.0158 & 0.0001 & 1.4228\tabularnewline
 & 5 &  & -0.0190 & 0.0001 & 1.1868 &  & -0.0191 & 0.0001 & 1.1807\tabularnewline
\hline 
\end{tabular}\caption{\label{table4}{\small{}The empirical study using the BMD dataset:
comparing the performances of different estimators in estimating the
prevalence of osteoporosis when ranking is done using the BMI category
($\rho=0.478$)}}
\end{table}

\begin{table}
\centering %
\begin{tabular}{cccccccccc}
\hline 
Sample Size & Set size &  & \multicolumn{3}{c}{MB estimator} &  & \multicolumn{3}{c}{ML estimator}\tabularnewline
\hline 
$n$ & $m$ &  & Bias & Variance & RE &  & Bias & Variance & RE\tabularnewline
 &  &  &  &  &  &  &  &  & \tabularnewline
 & 3 &  & -0.0096 & 0.0021 & 2.5487 &  & -0.0096 & 0.0021 & 2.5488\tabularnewline
10 & 4 &  & -0.0135 & 0.0016 & 3.1815 &  & -0.0136 & 0.0016 & 3.1926\tabularnewline
 & 5 &  & -0.0169 & 0.0012 & 3.6779 &  & -0.0171 & 0.0012 & 3.6921\tabularnewline
 &  &  &  &  &  &  &  &  & \tabularnewline
 & 3 &  & -0.0107 & 0.0010 & 2.5056 &  & -0.0107 & 0.0010 & 2.5062\tabularnewline
20 & 4 &  & -0.0148 & 0.0007 & 2.9318 &  & -0.0148 & 0.0007 & 2.9368\tabularnewline
 & 5 &  & -0.0180 & 0.0005 & 3.1002 &  & -0.0181 & 0.0005 & 3.1042\tabularnewline
 &  &  &  &  &  &  &  &  & \tabularnewline
 & 3 &  & -0.0108 & 0.0006 & 2.4007 &  & -0.0108 & 0.0006 & 2.4010\tabularnewline
30 & 4 &  & -0.0149 & 0.0004 & 2.6463 &  & -0.0150 & 0.0004 & 2.6478\tabularnewline
 & 5 &  & -0.0182 & 0.0003 & 2.6324 &  & -0.0183 & 0.0003 & 2.6321\tabularnewline
 &  &  &  &  &  &  &  &  & \tabularnewline
 & 3 &  & -0.0109 & 0.0005 & 2.2797 &  & -0.0109 & 0.0005 & 2.2802\tabularnewline
40 & 4 &  & -0.0152 & 0.0003 & 2.3682 &  & -0.0152 & 0.0003 & 2.3683\tabularnewline
 & 5 &  & -0.0186 & 0.0002 & 2.2443 &  & -0.0187 & 0.0002 & 2.2403\tabularnewline
 &  &  &  &  &  &  &  &  & \tabularnewline
 & 3 &  & -0.0110 & 0.0004 & 2.1884 &  & -0.0110 & 0.0004 & 2.1892\tabularnewline
50 & 4 &  & -0.0153 & 0.0002 & 2.1522 &  & -0.0154 & 0.0002 & 2.1507\tabularnewline
 & 5 &  & -0.0186 & 0.0002 & 1.9824 &  & -0.0187 & 0.0002 & 1.9766\tabularnewline
 &  &  &  &  &  &  &  &  & \tabularnewline
 & 3 &  & -0.0114 & 0.0002 & 1.7353 &  & -0.0114 & 0.0002 & 1.7351\tabularnewline
100 & 4 &  & -0.0155 & 0.0001 & 1.4815 &  & -0.0155 & 0.0001 & 1.4776\tabularnewline
 & 5 &  & -0.0187 & 0.0001 & 1.2279 &  & -0.0188 & 0.0001 & 1.2213\tabularnewline
\hline 
\end{tabular}\caption{\label{table5}{\small{}The empirical study using the BMD dataset:
comparing the performances of different estimators in estimating the
prevalence of osteoporosis when ranking is done using age category
($\rho=0.469$)}}
\end{table}

Table \ref{table4} presents the results when the ranking is done
using the BMIC. It can be observed from this table that the efficiency
gain using PNS estimators can be as large as $350\%$ in some certain
circumstances and never falls below one. For $n\leq30$, the REs of
both ML and MB estimators are increasing functions in $m$. For $n=40,$
the REs increase when we move from $m=3$ to $m=4$ and decrease when
we move from $m=4$ to $m=5$. For $n\geq50$, the REs are decreasing
functions in $m$. It should be noted that both estimators slightly
underestimate the true value of the prevalence of osteoporosis. It
is worth mentioning that both ML and MB estimators have competitive
performance in all the considered cases which is consistent with what
was observed in Section \ref{sec:Comparisons}.

Table \ref{table5} presents the results when AC is used as the ranking
variable. Although the REs are slightly higher in this table, their
patterns are very similar to those in Table \ref{table4}.

\subsection{Revisiting an earlier example}

Assume that the sample units in Table \ref{table1} are obtained from
the BMD dataset and we want to compare the estimators using the example
given in Section 2. Based on the data and the tie information in Table
\ref{table1}, we produce $F_{ml}\left(0.56\right)=0.0357$ and $F_{mb}\left(0.56\right)=0.0352$.
It can be observed that the ML estimator produces a slightly closer
value to the true value of the parameter of interest $\left(F\left(0.56\right)=0.061\right)$.

\section{Concluding remarks\label{sec:Concluding-Remarks}}

In this paper, a new cost-efficient sampling scheme was developed
for drawing statistical inference about either the lower or the upper
tail of the population distribution. The proposed procedure can be
applied in the situations in which measuring the variable of interest
is time-consuming, costly, and/or destructive. However, a small number
of the sampling units (set) can be ranked without taking the actual
measurements on them. In principle, the proposed procedure is similar
to nomination sampling (NS) with a clear modification which is made
to increase the applicability of this design in practice. In fact,
NS cannot be performed unless the researcher is able to determine
with a high confidence the sample unit with the lowest/highest rank
in the set . We proposed to modify NS in a way that the researcher
is allowed to declare as many ties as needed whenever he/she cannot
find with a high confidence the sample unit with the lowest/highest
rank in the set . In partial nomination sampling (PNS), the researcher
divides the sample units in the set into two subsets in a way that
the sample units in the first subset are all smaller than those in
the second one. However, the subset units in the second subset need
not be ranked. Finally, the researcher selects one unit at random
from the first (second) subset to obtain the Min(Max)PNS.

Then, two the cumulative distribution function (CDF) estimators were
developed based on the maximum likelihood (ML) and moment-based (MB)
approaches and the asymptotic normality of each of them was proved.
It was shown that under the perfect ranking assumption, the ML estimator
of the CDF was asymptotically more efficient than the MB estimator,
although the efficiency gain was rather small. For various choices
of sample size, set size, parent distribution, ties generation mechanism,
and quality of ranking, we compared the proposed estimators in PNS
with their counterparts in SRS using the Monte Carlo simulation. The
simulation results indicated that the developed estimators had very
competitive performances and were significantly more efficient than
their counterparts in simple random sampling (SRS) at either the lower
or the upper tail of the population distribution.

Finally, the proposed procedure was applied to a BMD dataset from
the Third National Health and Nutrition Examination Survey (NHANES
III) to estimate the prevalence of osteoporosis in adult women aged
50 and over. The ranking was done using the two inexpensive and easily
available ranking variables of body mass index and age category where
the ties naturally happened in the ranking process. It was observed
that MinPNS substantially improved the efficiency of the estimation
of osteoporosis. Therefore, it considerably reduced the time and cost
of the study. The findings of this paper encourage using the PNS design
to incorporate partial rank information into the estimation process. 

This work was the first but important attempt in developing statistical
inference about either the lower or the upper tail of the population
distribution using PNS. Therefore, it remains an ample space for future
research. For instance, the proposed methodology in this paper can
be used for estimating other population attributes rather than the
CDF. Let $\theta=g\left(F\right)$ be the parameter of interest where
$g\left(.\right)$ is an arbitrary function. This parameter can be
estimated in PNS by replacing the CDF of $F$ with the proposed estimators
in PNS. With the efficiency gain of the CDF estimators in PNS, it
is intuitively expected that the resulting estimators of $\theta$
will perform well in some certain circumstances. Moreover, note that
the asymptotic normality of the proposed methodology in this paper
requires the perfect ranking assumption which may not hold in some
practical situations. This, combined with the fact that \textit{a
large enough} sample size is needed to use the normal theory of the
estimators which may not be available in cost-efficient sampling designs
such as PNS, reminds us that using some alternative methods (such
as bootstrap techniques) is vital. These topics can be addressed in
subsequent works.

\section*{Data availability statement}

The data that support the findings of this study are obtained from
the third National Health and Nutrition Examination Survey (NHANES
III) and is available online at

\noindent\href{\%20http://www.cdc.gov/nchs/nhanes/nh3data.htm}{ http://www.cdc.gov/nchs/nhanes/nh3data.htm}.


\begin{thebibliography}{Jafari Jozani and Mirkamali(2010)}
\bibitem[Al-Omari and Haq(2011)]{Al-omari}Al-Omari, A. I., and Haq,
A., 2011, Improved quality control charts for monitoring the process
mean, using double-ranked set sampling methods. Journal of Applied
Statistics, 39 (4), 745-763.

\bibitem[Ahn et al.(2017)]{Ahn}Ahn, S., Wang, X., and Lim, J. 2017.
On unbalanced group sizes in cluster randomized designs using balanced
ranked set sampling, Statistics and Probability Letters, 123, 210-217.

\bibitem[Chen et al.(2005)]{Chen2005} Chen, H., Stasny, E. A., Wolfe,
D.A., 2005. Ranked set sampling for efficient estimation of a population
proportion. Statistics in Medicine 24, 3319-3329.

\bibitem[Chen et al.(2018)]{Chen2018}Chen, W., Tian, Y., Xie, M.
2018. The global minimum variance unbiased estimator of the parameter
for a truncated parameter family under the optimal ranked set sampling,
Journal of Statistical Computation and Simulation, 8 (17), 3399-3414

\bibitem[Chen et al.(2019)]{Chen2019}Chen, W., Yang, R., Yao, D.,
Long, C. 2019. Pareto parameters estimation using moving extremes
ranked set sampling, To appear in Statistical Papers.

\bibitem[Boyles and Samaniego(1986)]{Boyles}Boyles, R. A., and Samaniego,
F. J. 1986. Estimating a distribution function based on nomination
sampling, Journal of the American Statistical Association, 81 (396),
1039-1045.

\bibitem[Dell and Clutter(1972)]{Dell} Dell, T.R., and Clutter, J.L.,
1972. Ranked set sampling theory with order statistics background.
Biometrics 28, 2, 545-555.

\bibitem[Duembgen and Zamanzade(2020)]{Lutz} Duembgen, L., and Zamanzade,
2020. Inference on a distribution function from ranked set samples.
Annals of the institute of Statistical Mathematics, 72, 157\textendash 185.

\bibitem[Frey(2012)]{Frey2012} Frey, J. 2012. Nonparametric mean
estimation using partially ordered sets. Environmental and Ecological
Statistics, 19 (3), 309-326.

\bibitem[Frey and Wang(2013)]{Frey=000026Wang}Frey, J., and Wang,
L. 2013. Most powerful rank tests for perfect rankings, Computational
Statistics and Data Analysis, 60, 157-168.

\bibitem[Frey and Wang(2014)]{Frey-EDF}Frey, J., and Wang, L. 2014.
EDF-based goodness-of-fit tests for ranked-set sampling, Canadian
Journal of Statistics, 42 (3), 451-469.

\bibitem[Frey and Feeman(2016)]{FreyMean} Frey, J., and Feeman, T.G.
2016. Efficiency comparisons for partially rank-ordered set sampling.
Statistical Papers, 58 (4), 1149-1166.

\bibitem[Frey and Zhang(2017)]{Frey.perfect1}Frey, J., and Zhang,
Y. 2017. Testing perfect rankings in ranked-set sampling with binary
data. Canadian Journal of Statistics, 45 (3), 326-339.

\bibitem[Frey et al.(2007)]{Frey2007}Frey, J., Ozturk , O., and Deshpande,
J. V. 2007. Nonparametric tests for perfect judgment ranking, Journal
of the American Statistical Association, 102 (478), 708-717.

\bibitem[Frey and Zhang(2019)]{Frey2019}Frey, J. and Zhang, Y. 2019.
Improved exact confidence intervals for a proportion using ranked
set sampling, Journal of the Korean Statistical Society, 48 (3), 493-501.

\bibitem[Frey and Zhang(2021)]{Frey2021}Frey, J. and Zhang, Y. 2021.
Robust confidence intervals for a proportion using ranked-set sampling,
To appear in Journal of the Korean Statistical Society .

\bibitem[He et al.(2020)]{He2020}He, X., Chen, W., and Qian, W. 2020.
Maximum likelihood estimators of the parameters of the log-logistic
distribution, Statistical Papers, 61, 1875-1892.

\bibitem[He et al.(2021)]{He2021}He, X., Chen, W., and Rui, Y. 2021.
Modified best linear unbiased estimator of the shape parameter of
log-logistic distribution, Journal of Statistical Computation and
Simulation, 91 (2), 383-395.

\bibitem[Hjort and Pollard(2011)]{Hjort} Hjort, N.L., and Pollard,
D. 2011. Asymptotics for minimisers of convex processes. arXiv:1107.3806v1
{[}math.ST{]}

\bibitem[Johnell and Kanis(2006)]{Johnel} Johnell, O., and Kanis,
J. A. 2006. An estimate of the worldwide prevalence and disability
associated with osteoporotic fractures. Osteoporos Int 17:1726.

\bibitem[Jafari Jozani and Mirkamali(2010)]{Jozani-Mirkamali1}Jafari
Jozani, M. and Mirkamali, S. J. 2010. Improved attribute acceptance
sampling plans based on maxima nomination sampling. Journal of Statistical
Planning and Inference. 140, 2448-2460.

\bibitem[Jafari Jozani and Mirkamali(2011)]{Jozani-Mirkamali2}Jafari
Jozani M, Mirkamali S. J. 2011. Control charts for attributes with
maxima nominated samples. Journal of Statistical Planning and Inference;141:2386-2398.

\bibitem[Jafari Jozani et al.(2018)]{Jozani-NSreg}Jafari Jozani,
M., Ayilara, O. F., and Leslie, W. B. 2018. Quantile regression with
nominated samples: An application to a bone mineral density study
. Statistics in Medicine, 37 (14), 2267-2283.

\bibitem[Haq et al.(2013)]{Haq} Haq, A., Brown, J., Moltchanova,
E., and Al-Omari, A. I. 2013. Partial ranked set sampling design,
Environmetrics, 24 (3), 201-207.

\bibitem[Haq and Al-Omari(2014)]{Haq-almori1}Haq, A., and Al-Omari,
A. I., 2014, A new Shewhart control chart for monitoring process mean
based on partially ordered judgment subset sampling, Quality and Quantity,
49 (3), 1185-1202.

\bibitem[Haq et al.(2014)]{Haqetal2}Haq, A., Brown, J., Moltchanova,
E., and Al-Omari, A. 2014. Effect of measurement error on exponentially
weighted moving average control charts under ranked set sampling schemes,
Journal of Statistical Computation and Simulation, 85 (6), 1224-1246.

\bibitem[Kanis et al.(2000)]{Kanis} Kanis J A, Johnell O, Oden A,
et al. 2000. Long-term risk of osteoporotic fracture in Malmo. Osteoporos
Int 11:669.

\bibitem[Melton et al. (1998)]{Melton1998}Melton L J, Atkinson E
J, O'Connor MK, et al. 1998. Bone density and fracture risk in men.
J Bone Miner Res 13:1915.

\bibitem[Melton et al. (1992)]{Melton1992}Melton L J, Chrischilles
E A, Cooper C, et al. 1992. Perspective. How many women have osteoporosis?
J Bone Miner Res 7:1005.

\bibitem[McIntyre(1952)]{McIntyre} McIntyre, G.A., 1952. A method
for unbiased selective sampling using ranked set sampling. Austral.
J. Agricultural Res. 3, 385-390.

\bibitem[Mu(2015)]{Mu} Mu, X., 2015. Log-concavity of a mixture of
beta distributions, Statistics and Probability Letters, 99, 125-130

\bibitem[Nourmohammadi et al.(2014)]{Nourmohammadi2014}Nourmohammadi,
M., Jafari Jozani, M. and Johnson, B. 2014. Confidence interval for
quantiles in finite populations with randomized nomination sampling.
Computational Statistics and Data Analysis. 73, 112\textendash 128.

\bibitem[Omidvar et al.(2018)]{Omidvar}Omidvar, S., Jafari Jozani,
M. and Nematollahi, N. 2018. Judgment post-stratification in finite
mixture modeling: An example in estimating the prevalence of osteoporosis,
Statistics in Medicine, 37 (30), 4823-4836

\bibitem[Ozturk(2011)]{Ozturk2011} Ozturk, O. 2011. Sampling from
partially rank-ordered sets. Environmental and Ecological Statistics,
18, 757-779.

\bibitem[Qian et al.(2021)]{Qian2021}Qian, W., Chen, W., He, X. 2021.
Parameter estimation for the Pareto distribution based on ranked set
sampling. Statistical Papers, 62, 395-417.

\bibitem[Samawi and Al-Sagheer(2001)]{Samawi-CDF-2001}Samawi, H.
M., and Al-Sagheer, O. A. 2001. On the estimation of the distribution
function using extreme and median ranked set samples. Biometrical
Journal, 43 (3), 357-373.

\bibitem[Samawi and Al-Saleh(2013)]{Samawi-odds-2013}Samawi, H. M.,
and Al-Saleh, M.F. 2013. Valid estimation of odds ratio using two
types of moving extreme ranked set sampling. Journal of the Korean
Statistical Society 42, 17-24.

\bibitem[Samawi et al.(2017)a]{Samawi-logi-2017}Samawi, H. M., Rochani,
H., Linder, D., and Chatterjee, A. 2017. More efficient logistic analysis
using moving extreme ranked set sampling. Journal of Applied Statistics,
44 (4), 753-76.

\bibitem[Samawi et al.(2017)b]{Samawi-StatMed-2017}Samawi, H. M.,
Yin, J., Rochani, H., Panchal, V. 2017. Notes on the overlap measure
as an alternative to the Youden index: How are they related? Statistics
in Medicine, 36 (26), 4230-4240.

\bibitem[Samawi et al.(2018)]{Samawi-sample-2018}Samawi, H. M., Helu,
A., Rochani, H., Yin, J., Yu, L., and Vogel, R. 2018. Reducing sample
size needed for accelerated failure time model using more efficient
sampling methods. Statistical Theory and Practice, 12 (3), 530-541.

\bibitem[Si et al. (2015)]{Si}Si, L., Winzenberg, T. M., Jiang, Q.,
Chen, M., Palmer, A. J. 2015.Projection of osteoporosis-related fractures
and costs in China: 2010\textendash 2050. Osteoporos Int 26, 1929-1937.

\bibitem[Stokes and Sager(1988)]{Stokes}Stokes, S. L., and Sager,
T. W. Characterization of a ranked-set sample with application to
estimating distribution functions. Journal of the American Statistical
Association, 83 (402), 374-381.

\bibitem[Takahasi and Wakitomo(1968)]{Takahasi}Takahasi, K., and
Wakimoto, K. 1968. On unbiased estimates of the population mean based
on the sample stratified by means of ordering. Annals of the Institute
of Statistical Mathematics, 20(1), 1-31

\bibitem[Wang et al.(2012)]{Wang2012} Wang, X., Wang, K., and Lim,
J. 2012. Isotonized CDF estimation from judgment poststratification
data with empty strata. Biometrics, 61 (1), 194-202.

\bibitem[Wang et al.(2016)]{Wang2016}Wang, X., Lim, J., and Stokes,
S. L. 2016. Using ranked set sampling with cluster randomized designs
for improved inference on treatment effects. Journal of the American
Statistical Association, 111 (516), 1576-1590.

\bibitem[Wang et al.(2017)]{Wang2017}Wang, X., Ahn, S., and Lim,
J. 2017. Unbalanced ranked set sampling in cluster randomized studies.
Journal of Statistical Planning and Inference, 187, 1-16.

\bibitem[Wang et al.(2020)]{Wang2020}Wang, X., Wang, M., Lim, J.,
and Ahn, S. 2020. Using ranked set sampling with binary outcomes in
cluster randomized designs. Canadian Journal of Statistics, 48 (3),
342-365.

\bibitem[Willemain(1980)a]{Willemain-a}Willemain, T. 1980a. A comparison
of patient-centered and case-mix reimbursement for nursing home care,
Health Service Research, 15 (4), 365-377.

\bibitem[Willemain(1980)b]{Willemain-b}Willemain T. 1980. Estimating
the population median by nomination sampling. Journal of American
Statistical Association.75 (372):908-911.

\bibitem[Zamanzade and Mahdizadeh(2017)]{Zamanzade2017} Zamanzade,
E. and Mahdizadeh, M. 2017. A more efficient proportion estimator
in ranked set sampling. Statistics and Probability Letters, 129, 28-33.

\bibitem[Zamanzade and Mahdizadeh(2020)]{Zamanzade=000026Mahi2020}
Zamanzade, E., and Mahdizadeh, M. 2020. Using ranked set sampling
with extreme ranks in estimating the population proportion, Statistical
Methods in Medical Research, 29 (1), 165-177.
\end{thebibliography}
\end{document}